%% file: Sparse-MIMO-Radar.tex
\documentclass[10pt,a4paper]{article}

\usepackage[a4paper, left=2.8cm, right=2.8cm, top=4.3cm, bottom=4.1cm]{geometry}
\usepackage[fleqn]{amsmath}
\usepackage{amssymb}
\usepackage{amsthm}
\usepackage{enumerate}
\usepackage[x11names]{xcolor}
\usepackage{tikz}
\usepackage{hyperref}
\hypersetup{
  colorlinks = true,
  linkcolor = red,
  anchorcolor = red,
  citecolor = red,
  filecolor = red,
  urlcolor = gray
}
\usepackage{fullpage}
\usepackage{microtype}

\usepgflibrary{plotmarks}

\newtheorem{definition}{Definition}
\newtheorem{theorem}[definition]{Theorem}
\newtheorem{lemma}[definition]{Lemma}
\newtheorem{corollary}[definition]{Corollary}

\newtheorem{proposition}[definition]{Proposition}
\newtheorem{remark}[definition]{Remark}

\include{latex_inc}

\newcommand{\vxTildeArg}[1]{\widetilde{\vc{V}}_{\hspace{-.3em} \vc{#1} }}
\newcommand{\vxTilde}{\vxTildeArg{x}}
\newcommand{\x}[1]{\boldsymbol{\XX}_{\hspace{-.1em}#1}}
\newcommand{\ee}[1]{e^{\imath 2\pi \cdot #1}}
\newcommand{\nee}[1]{e^{-\imath 2\pi \cdot #1}}
\newcommand{\vc}[1]{\boldsymbol{ #1 }}
\newcommand{\opnorm}[1]{\| #1 \|_{2\to 2}}
\newcommand{\betacl}{[\beta]}
\newcommand{\ID}{\boldsymbol{\id}}

\newcommand{\midcol}{{}:{}}

\begin{document}

\title{\bf Refined analysis of sparse MIMO radar}

\author{
Dominik Dorsch
\thanks{
Chair for Mathematics C (Analysis),
RWTH Aachen University,
\href{mailto:dorsch@mathc.rwth-aachen.de}{dorsch@mathc.rwth-aachen.de}
}
\and
Holger Rauhut\thanks{
Chair for Mathematics C (Analysis),
RWTH Aachen University,
\href{mailto:rauhut@mathc.rwth-aachen.de}{rauhut@mathc.rwth-aachen.de}
}
}
\date{\today}

\maketitle

\begin{abstract}
We analyze a multiple-input multiple-output (MIMO) radar model and provide recovery results for a compressed sensing (CS) approach.
In MIMO radar different pulses are emitted by several transmitters and the echoes are recorded at several receiver nodes.
Under reasonable assumptions the transformation from emitted pulses to the received echoes 
can approximately be regarded as linear.
For the considered model, and many radar tasks in general, sparsity of targets within the considered angle-range-Doppler domain is a natural assumption.
Therefore, it is possible to apply methods from CS in order to reconstruct the parameters of the targets.
Assuming Gaussian random pulses the resulting measurement matrix becomes a highly structured random matrix.
Our first main result provides an estimate for the well-known restricted isometry property (RIP) ensuring stable and robust recovery.
We require more measurements than standard results from CS, like for example those for Gaussian random measurements.
Nevertheless, we show that due to the special structure of the considered measurement matrix our RIP result is in fact optimal (up to possibly logarithmic factors).
Our further two main results on nonuniform recovery (i.e., for a fixed sparse target scene) reveal how the fine structure of the support set --- not only the size --- affects the (nonuniform) recovery performance.
We show that for certain {\lq\lq balanced\rq\rq} support sets reconstruction with essentially the optimal number of measurements is possible.
Indeed, we introduce a parameter measuring the well-behavedness of the support set and resemble standard results from CS for near-optimal parameter choices.
We prove recovery results for both perfect recovery of the support set in case of exactly sparse vectors and an $\ell_2$-norm approximation result for reconstruction under sparsity defect.
Our analysis complements earlier work by \mbox{Strohmer \& Friedlander} and deepens the understanding of the considered MIMO radar model.
\end{abstract}



\section{Introduction}
MIMO (multiple-input multiple-output) radar systems can simultaneously transmit several uncorrelated waveforms from spatially distributed transmitters and record the reflected signals at different receiver locations.
The transmit/receive antennas may either be widely separated giving the possibility to view the targets from different angles or co-located antenna configurations --- to be studied in this paper --- can yield superior resolutions and target identifiability as compared to standard phased-array \mbox{radar \cite{Fishler2004,Stoica2007,Stoica2009}}.

When it comes to the implementation of a MIMO radar system involving several waveforms, one is naturally confronted with an increased amount of data to be processed.
Yet in MIMO radar the description of the target parameters can turn out to be particularly sparse, i.e., the target scene can be modeled as a vector of the object reflectivities with the property that {\lq\lq almost\rq\rq} all entries are zero \cite{Chen2008}. 
Since it is often possible to assume the transformation of the transmitted signals through the channel to be approximately linear, one is faced with the reconstruction of the targets' parameters from highly incomplete linear measurements.

During the last decade, the theory of compressed sensing (CS) evolved around this particular setting.
In CS one typically tries to reconstruct an unknown vector $\vc{x}$ from highly underdetermined linear measurements $\vc{y} = \vc{A} \vc{x}$ under the additional assumption of sparsity in $\vc{x}$.
While only suboptimal results are available so far for deterministic measurement matrices $\vc{A}$ --- see for instance
the discussion in \cite[Chapter~6.1]{Foucart2013} ---  
it is rather typical to consider matrices which are defined in terms of random variables.
The crucial point is that, apart from theory, also in practice the concept of randomness can be implemented efficiently and has already  proven effective.
This opens the field for advanced methods from probability theory making it possible to prove strong recovery results --- which usually hold true with very high probability.
Typical results guaranteeing recovery of sparse vectors $\vc{x}$ from measurements $\vc{y}=\vc{A} \vc{x}$
provide conditions on the minimal number of required measurements --- ideally scaling 
linearly on the number of nonzero entries \mbox{in $\vc{x}$} up to logarithmic factors.
Given this, the theory of CS provides several algorithms for reconstruction.
Usually $\vc{x}$ can either be obtained as the solution to a convex optimization problem, or by an iterative (greedy) 
algorithm \cite{Foucart2013}.

Ideas from CS have probably first been applied to radar in \cite{bast07,Strohmer2008a,en10-1}.
The extension to MIMO radar has been conducted in \cite{Chen2008,Strohmer2009,Yu2011,Yu2012}.
In a recent paper, \mbox{Strohmer \& Friedlander \cite{Strohmer2012}} analyze a particular model of 
co-located MIMO radar and prove recovery results.
They assume random probing signals which leads to a measurement process constituting of $N_R N_t$ random measurements, where $N_R$ is the number of receiver nodes and $N_t$ is the number of (time domain) samples being taken at each receiver.
Under the additional assumption of randomly distributed targets, they show that basically the condition $N_R N_t \gtrsim s \log (N)$, where $s$ is the number of targets and $N$ is the dimension of the target scene $\vc{x}$, 
is sufficient for reconstruction by minimizing the so-called LASSO functional.
Thus, typical results from CS on random measurements are resembled for these specific MIMO radar measurements.

We adopt this model and deepen its understanding by deriving further properties and results in the context of CS.
First, we prove an estimate for the restricted isometry property (RIP) of the involved measurement matrix.
The RIP guarantees uniform recovery (i.e., simultaneous recovery of all sufficiently sparse vectors from a single random draw of the matrix with high probability) and, 
moreover, yields strong error estimates for noisy measurements and only approximately sparse vectors.
Compared to standard estimates for standard random matrix constructions, we require more measurements for a given sparsity which, however, as we show, is still optimal and due to the special structure of the considered matrix.
Motivated by the somehow better results in \cite{Strohmer2012}, we furthermore deepen the analysis of the measurement process by introducing a parameter for the fine structure of the supports sets.
In this way we are able to provide nonuniform recovery results for {\lq\lq balanced\rq\rq} support sets resembling the required number of measurements one would usually expect from common CS theory.
Ultimately, this explains the result in \cite{Strohmer2012}, since random support sets are --- on average --- sufficiently balanced and, hence, for randomly distributed targets less measurements are sufficient.
To the best of our knowledge, it has not been observed earlier for realistic measurement matrices that the recovery performance may depend significantly on the fine structure of the support sets.

\subsection{The MIMO radar model}
Our MIMO radar model consists of $N_T$ transmit and $N_R$ receive antennas.
We consider the setting from \cite{Strohmer2012} where some assumptions on the geometry of the antenna/target locations have been made.
We assume that
\begin{itemize}
 \item the antenna arrays and the scatterers are located in the same two-dimensional plane,
 \item the transmit/receive antennas are located along one common line (\lq\lq monostatic radar\rq\rq),
 \item the distance of the targets from the antenna arrays is sufficiently large so that the radar return of any scatterer can be considered to be fully correlated across the array (\lq\lq coherent propagation scenario\rq\rq).
\end{itemize}

We point out that our results can be generalized to a three-dimensional setting.
However, since the effects we want to study already occur in two dimensions, we concentrate on this case in order to keep the notation simple.
The transmit antennas occupy the positions $(0, i d_T \lambda ) \in \R^2$, $i=0, 1, \ldots, N_T - 1$, where $\lambda$ denotes the wavelength of the carrier frequency of the radar system.
The receive antennas are located at $(0, j d_R \lambda )$, $j=0, 1, \ldots, N_R - 1$.
It is known that by choosing $d_T = 1/2$ and $d_R = {N_T}/{2}$ or, alternatively, $d_T = {N_R}/{2}$ and $d_R = {1}/{2}$ similar characteristics as those of a virtual array with $N_T N_R$ antennas can be obtained \cite{Friedlander2009}.
In this sense these particular choices for $d_T$ and $d_R$ are favorable compared to other choices in practice.
This fact will also become clear during our analysis. 
Throughout the paper we concentrate on the case where 
\begin{equation}\label{eqn:dT_dR}
 d_T = \frac{1}{2}, \qquad d_R = \frac{N_T}{2} .
\end{equation}
The second case can be treated analogously.

The $i$th transmit antenna repeatedly sends a fixed complex continuous-time signal $s_i = s_i (t)$ of period duration $T$.
The reflected signals, due to reflections caused by one unit reflectivity target at position $\big( r \cos ( \theta ) , r \sin ( \theta ) \big) \in \R^2$, traveling with radial speed $v$, will be recorded at the $j$th receiver (cf. \cite{Yu2012}) as
\[
 y_j (t) = \sum_{i=1}^{N_T} \nee{ c \lambda^{-1} ( t - {d_{i,j} (t)}/{c} ) } s_i ( t - {d_{i,j} (t)}/{c} ) ,
\]
where $c$ denotes the speed of light and $d_{i,j} (t)$ is the distance the emitted wave has to travel from the $i$th transmitter to the $j$th receiver at time $t$, given by
\[
 d_{i,j} (t) = 2 ( r + v t ) + \sin ( \theta ) d_T \lambda (i-1) + \sin ( \theta ) d_R \lambda (j-1) .
\]
After demodulation (multiplication with $\ee{c\lambda^{-1} t}$), assuming narrowband transmit waveforms, slowly moving targets, and a far field scenario, i.e., $r \gg \max \{ d_T N_T \lambda , d_R N_R \lambda \}$, the received baseband signal at the $j$th receiver is approximately given by
\begin{equation}\label{eqn:received_sig}
 y_j (t) \approx \ee{2 \lambda^{-1} r} \ee{\sin ( \theta ) d_R (j-1)} \sum_{i=1}^{N_T} \ee{2 \lambda^{-1} v t} \ee{\sin ( \theta ) d_T (i-1)} s_i ( t - {2r}/{c} ) .
\end{equation}
(See also \cite{Yu2012}, where the antennas are distributed freely over a small area, though.)

The parameters of a target are given by the triple $( \theta , r , v )$ representing the position (in radial coordinates) and the radial speed.
It will, however, be more convenient in the remaining part of this paper to equivalently consider the triples $( \sin ( \theta ) , 2r/c , 2\lambda^{-1} v )$ of angle, delay, and Doppler-shift parameters, which appear in formula \eqref{eqn:received_sig}.
The angle, delay, and Doppler-shift information can always easily be transformed back into the physical coordinates.

\begin{figure}[t]
\begin{center}
\begin{tikzpicture}[scale=0.42]

\foreach \r in {10, 11, ..., 17}
  \draw[dotted,opacity=0.7] ([shift={(-23:\r)}]0,0) arc (-23:23:\r);
\foreach \a in {-69.636,-61.045,-54.341,-48.59,-43.433,-38.682,-34.229,-30,-25.944,-22.024,-18.21,-14.478,-10.807,-7.1808,-3.5833,0,3.5833,7.1808,10.807,14.478,18.21,22.024,25.944,30,34.229,38.682,43.433,48.59,54.341,61.045,69.636,90}
  \draw[dotted,Azure4] (\a:3.2) -- (\a:3.8);
\foreach \a in {-48.59,-30,-14.478,0,14.478,30,48.59,90}
  \draw[thick,SteelBlue3] (\a:3) -- (\a:4);
\foreach \a in {-30,-14.478,0,14.478,30}
  \draw[dotted,SteelBlue3] (\a:4) -- (\a:10);
\foreach \a in {-18.21,-14.478,-10.807,-7.1808,-3.5833,0,3.5833,7.1808,10.807,14.478,18.21}
  \draw[dotted,Azure4] (\a:10) -- (\a:17);
\foreach \a in {-14.478,0,14.478}
  \draw[thick,SteelBlue3] (\a:10) -- (\a:17);

\foreach \ttau in {0,1,...,7}
  \draw[opacity=0.7] (\ttau + 10,0) node[inner sep=1pt,below=3pt,rectangle,fill=white] {$\tau_{\ttau}$};
\draw[SteelBlue3,opacity=1] (-48.59: 4.7) node {$\beta_{8}$};
\draw[SteelBlue3,opacity=1] (-30: 10.7) node {$\beta_{16}$};
\draw[Azure4,opacity=0.4] (-18.21: 17.7) node {$\beta_{22}$};
\draw[SteelBlue3] (-14.478: 17.7) node {$\beta_{24}$};
\draw[Azure4] (-10.807: 17.7) node {$\beta_{26}$};
\draw[Azure4] (-7.1808: 17.7) node {$\beta_{28}$};
\draw[Azure4] (-3.5833: 17.7) node {$\beta_{30}$};
\draw[SteelBlue3] (0: 17.7) node {$\beta_{32}$};
\draw[Azure4] (3.5833: 17.7) node {$\beta_{34}$};
\draw[Azure4] (7.1808: 17.7) node {$\beta_{36}$};
\draw[Azure4] (10.807: 17.7) node {$\beta_{38}$};
\draw[SteelBlue3] (14.478: 17.7) node {$\beta_{40}$};
\draw[Azure4,opacity=0.4] (18.21: 17.7) node {$\beta_{42}$};
\draw[SteelBlue3,opacity=1] (30: 10.7) node {$\beta_{48}$};
\draw[SteelBlue3,opacity=1] (48.59: 4.7) node {$\beta_{56}$};
\draw[SteelBlue3,opacity=1] (90: 4.7) node {$\beta_{64}$};

\draw (-4.9,0) node[rounded corners=1pt, draw, fill=Azure3]{\bf MIMO radar module};
\foreach \r in {0.5, 1, 1.5}
  \draw[opacity=0.6,dashed] ([shift={(-50:\r)}]0,0) arc (-50:50:\r);
  
\draw[->,dashed,thick,Goldenrod1] (-10.807:16) -- (-10.807:13.5);
\draw[Goldenrod1] (-10.807:13) node {$f_k$};

\draw[draw=red,fill=red] (3.5833:13) circle(1.7mm);
\draw[draw=red,fill=red] (14.478:11) circle(1.7mm);
\draw[draw=red,fill=red] (-10.807:16) circle(1.7mm);
\draw[draw=red,fill=red,opacity=0.4] (-18.21:14) circle(1.7mm);
\draw[draw=red,fill=red,opacity=0.4] (18.21:16) circle(1.7mm);
\draw[draw=red,fill=red,opacity=0.08] (22.024:17) circle(1.7mm);
\draw[red] (12,1.8) node {\bf sparse target scene};
\end{tikzpicture}
\end{center}
\caption{
Sparse MIMO radar: sector of the physical target domain for the case of $N_T = 8$ and $N_R=8$ transmit/receive antennas, yielding an angular resolution of $N_T N_R = 64$.
An exemplary delay grid with a resolution of $N_{\tau} = 8$ is depicted.
The considered Doppler effect introduces a third dimension, parametrized by the Doppler shifts $f$.
}
\label{radar_fig}
\end{figure}
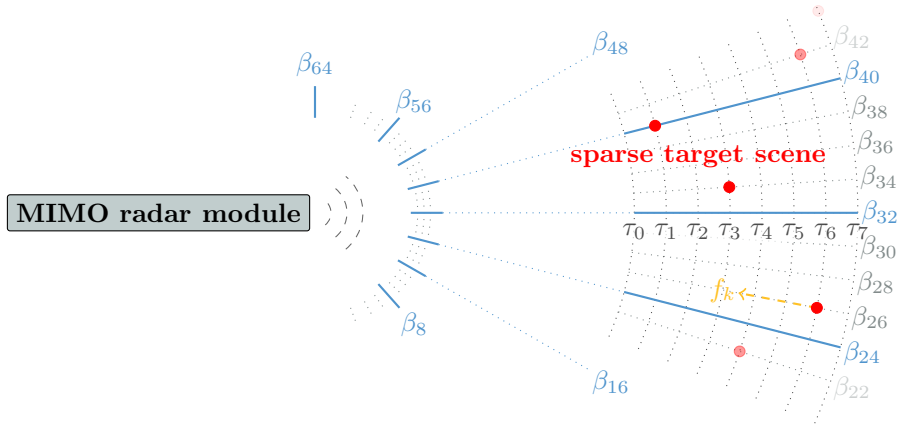

\subsubsection*{Discretization}
We switch to a discrete setting by considering sampled (at Nyquist rate) versions of the transmitted signals and, furthermore, assuming that the targets are located on a grid in the angle-delay-Doppler domain.
In certain applications, this idealizing assumption might lead to so-called {\lq\lq gridding errors\rq\rq}.
Thus, in practice, this would have to be resolved by an additional post-processing.
Theoretical approaches for this problem can be found in \cite{Herman2010,Chi2011}.
Recently, in \cite{Heckel2014}, a new approach for exact reconstruction of time/frequency shifts (as considered in radar) from a continuous domain has been proposed.

Let $\vc{s}_i \in \C^{N_t}$ be the discrete-time representation of the transmitted signal $s_i = s_i (t)$ sampled at Nyquist-rate over the time interval $[0,T)$, i.e.,
\[
 \vc{s}_i^T = ( s_i(0) , s_i(\Delta_t) , s_i(2\Delta_t) , \ldots , s_i( (N_t -1) \Delta_t))^T , \qquad  T=N_t \Delta_t , \qquad \Delta_t = 1 / 2B.
\]
Under the assumption that the signals $s_i = s_i (t)$ are $T$-periodic and band-limited to $(-B,B)$ the Nyquist--Shannon sampling theorem tells us that they are fully described by the discrete-time signals $\vc{s}_i$.
In the following we will always assume that the parameters $( \sin ( \theta ) , 2r/c , 2\lambda^{-1} v )$ lie on an equidistant grid, i.e.,
\begin{equation}\label{eqn:parameters_on_grid}
 ( \sin ( \theta ) , 2r/c , 2\lambda^{-1} v ) = ( \beta \Delta_\beta , \tau \Delta_\tau , f \Delta_f ),
\end{equation}
where $\beta$, $\tau$, and $f$ are integers.
Throughout the rest of this paper we fix the stepsizes $\Delta_\beta$, $\Delta_\tau$, and $\Delta_f$ to
\begin{equation}\label{eqn:stepsizes}
 \Delta_\beta = \frac{2}{N_T N_R}, \qquad \Delta_\tau = \frac{1}{2B}, \qquad \Delta_f = \frac{1}{T}.
\end{equation}
For a given Doppler-shift $2\lambda^{-1} v = f \Delta_f \ll B$, the received signal $y_j = y_j (t)$ can again be assumed to be band-limited to $(-B,B)$.
In this case the time-discrete version
\[
 \vc{y}_j^T = ( y_j(0) , y_j(\Delta_t) , y_j(2 \Delta_t) , \ldots , y_j((N_t -1) \Delta_t))^T \in \C^{N_t}
\]
of the approximate received signal $y_j = y_j (t)$ in \eqref{eqn:received_sig} at the $j$th receiver, caused by a target having parameters as in \eqref{eqn:parameters_on_grid}, is given by the vector
\begin{equation}\label{eqn:discrete_reveived_signal}
 \vc{y}_j = \ee{c \lambda^{-1} \tau \Delta_\tau} \ee{d_R \beta \Delta_\beta (j-1)} \sum_{i=1}^{N_T} \ee{d_T \beta \Delta_\beta (i-1)} \vc{M}_f \vc{T}_\tau \vc{s}_i \in \C^{N_t} ,
\end{equation}
where $\vc{T}_\tau$ is a circulant shift and $\vc{M}_f$ is a linear modulation, defined entrywise by
\begin{equation}\label{eqn:signal_ops_def}
 [ \vc{T}_\tau \vc{s} ]_k = [ \vc{s} ]_{k - \tau}, \qquad [ \vc{M}_f \vc{s} ]_k = \ee{\frac{f}{N_t}(k-1)} [ \vc{s} ]_k ,
\end{equation}
where $k - \tau$ stands for subtraction modulo $N_t$.
Note that, due to periodicity of the operators $\vc{T}_\tau$ and $\vc{M}_f$ and the periodic influence of the parameter $\beta$ in \eqref{eqn:discrete_reveived_signal}, the maximal achievable number of grid points as in \eqref{eqn:parameters_on_grid} which can be differentiated by looking at the received signals $\vc{y}_j$ is bounded by $N_T N_R N_t^2$.
Hence we assume without loss of generality that 
\begin{equation}\label{eqn:grid}
 ( \beta , \tau , f ) \in \GG, \qquad \GG := [N_T N_R] \times [N_t] \times [N_t] ,
\end{equation}
where, for any $L\in\N$, we write $[L]$ to denote the set $\{ 1 , 2 , \ldots , L \}$.
Throughout this paper we will refer to the triples of integers $(\beta,\tau,f)$ as \textit{parameters} of the considered targets keeping in mind that the actual \textit{physical} parameters can be obtained via the relations \eqref{eqn:parameters_on_grid}, \eqref{eqn:stepsizes}.

\subsubsection*{The linear measurement model}
If, more generally, several targets $\varTheta_k = ( \beta_k , \tau_k , f_k )$, $k=1, 2, \ldots, L$, with corresponding reflectivities $\rho_k \in \C$ are present, then the superposition of signals
\[
 \vc{y}_j = \sum_{k=1}^L \rho_k \cdot \ee{c \lambda^{-1} \tau_k \Delta_\tau} \bigg[ \ee{d_R \beta_k \Delta_\beta (j-1)} \sum_{i=1}^{N_T} \ee{d_T \beta_k \Delta_\beta (i-1)} \vc{M}_{f_k} \vc{T}_{\tau_k} \vc{s}_i \bigg] \in \C^{N_t}
\]
will be recorded at the $j$th receiver, see also \eqref{eqn:discrete_reveived_signal}.
We define vectors $\vc{A}_{\varTheta}^{j} \in \C^{N_t}$ with
\begin{equation}\label{eqn:columns}
 \vc{A}_{\varTheta}^j = \ee{d_R \beta \Delta_\beta (j-1)} \sum_{i=1}^{N_T} \ee{d_T \beta \Delta_\beta (i-1)} \vc{M}_{f} \vc{T}_{\tau} \vc{s}_i , \qquad \varTheta = (\beta , \tau , f) \in \GG ,
\end{equation}
which allows us to write $\vc{y}_j$ more conveniently as the sum $\sum_{k=1}^L \rho_k \ee{c \lambda^{-1} \tau_k \Delta_\tau} \vc{A}_{\varTheta_k}^j$.
A \textit{target scene} (a set of targets being present in the angle-delay-Doppler domain) can now equivalently be considered as a vector \mbox{$\vc{x} = ( x_\varTheta , \, \varTheta \in \GG ) \in \C^{N}$}, $N = N_T N_R N_t^2$, being supported on the indices $\varTheta = \varTheta_k$, $k=1, 2, \ldots, L$, which correspond to the parameters $\varTheta_k$ of the targets.
Here, each nonzero entry is basically the reflectivity parameter of the corresponding target (multiplied with a complex unit), given by $x_{\varTheta_k} = \rho_k \ee{c \lambda^{-1} \tau_k \Delta_\tau}$.
The collection of all received signals $(\vc{y}_1^T , \vc{y}_2^T , \ldots , \vc{y}_{N_R}^T )$ can now conveniently be written as matrix--vector product $\vc{A}\vc{x}$, where the matrix $\vc{A} \in \C^{N_R N_t \times N}$ contains the columns $\vc{A}_\varTheta$, $\varTheta \in \GG$, with
\begin{equation}\label{eqn:columns2}
 \vc{A}_\varTheta = ( ( \vc{A}_{\varTheta}^{1} )^T , ( \vc{A}_{\varTheta}^{2} )^T , \ldots , ( \vc{A}_{\varTheta}^{N_R} )^T  )^T \in \C^{N_R N_t} .
\end{equation}
Considering noise in the channel leads to the measurement model
\begin{equation}\label{eqn:measurements}
 \vc{y} = \vc{Ax} + \vc{n} \in \C^{N_R N_t},
\end{equation}
where $\vc{n}$ represents a noise vector.
Since $\vc{x}$ is an $N = N_R N_T N_t^2$-dimensional vector and the acquired information in the measured vector $\vc{y}$ is $m = N_R N_t$-dimensional, we end up with a highly under-determined system \eqref{eqn:measurements} from which the targets' parameters (the support set of the vector $\vc{x}$) have to be reconstructed.

\subsection{Recovery via \texorpdfstring{$\ell_1$-minimization}{l1-minimization} and the restricted isometry property}
Radar scenes are typically sparse in the target domain meaning $s = |\text{supp} (\vc{x})| \ll N$.
One can use this additional assumption for the reconstruction of $\vc{x}$.
At this point ideas from compressed sensing enter the field.
A well established --- and computationally practicable --- approach for reconstructing a sparse vector $\vc{x}$ from incomplete measurements $\vc{y}$ is \textit{basis pursuit denoising} which aims at reconstructing $\vc{x}$ by solving the convex optimization problem
\begin{equation}\label{eqn:basis_pursuit}
 \min_{\vc{z}} \| \vc{z} \|_1 \quad \text{subject to} \quad \| \vc{Az} - \vc{y} \|_2 \leq \varrho ,
\end{equation}
where $\varrho$ is an upper estimate of the noise level.
Here the $\| \cdot \|_1$-norm constitutes a convex relaxation of the $\| \cdot \|_0$-norm counting the number of nonzero entries, which one aims to minimize in the first place.
Indeed, under suitable assumptions on the measurement matrix $\boldsymbol{A}$, each solution $\vc{x}^\#$ of \eqref{eqn:basis_pursuit}
is close to the unknown $s$-sparse vector $\vc{x}$, i.e.,
\begin{align}\label{eqn:approx}
\| \vc{x}^\# - \vc{x} \|_1 & \leq C \inf_{\text{$s$-sparse $\vc{x}'$}} \| \vc{x}' - \vc{x} \|_1 + D\varrho\sqrt{s}, \\
\| \vc{x}^\# - \vc{x} \|_2 & \leq C \frac{\inf_{\text{$s$-sparse $\vc{x}'$}} \| \vc{x}' - \vc{x} \|_1}{\sqrt{s}} + D\varrho,  \label{eqn:approx2}
\end{align}
where $C,D > 0$ are numerical constants.

A popular condition guaranteeing recovery via basis pursuit denoising is the \textit{restricted isometry property (RIP)}.
The RIP is said to be fulfilled if $\boldsymbol{A}$ has a small restricted isometry constant which is the smallest number $\delta_s$ such that for all $s$-sparse vectors $\vc{x}$ it holds
\begin{equation}\label{eqn:rip}
 (1-\delta_s) \|\vc{x}\|_2^2 \leq \| \vc{A} \vc{x}\|_2^2 \leq (1+\delta_s) \|\vc{x}\|_2^2 .
\end{equation}
A sufficiently small restricted isometry constant implies that the basis pursuit denoising approach \eqref{eqn:basis_pursuit} yields approximations to any $s$-sparse vector $\vc{x}$ from incomplete measurements $\vc{y}$.
The following theorem is well-known \cite{cazh14}, see for instance also \cite{carota06-1,ca08,Foucart2013}, where constants have been subsequently been improved.

\begin{theorem}\label{thm:rip_reconstr_result}
Suppose the restricted isometry constant $\delta_{2s}$ of the matrix $\vc{A}$ satisfies $\delta_{2s} < 1 / \sqrt{2}$.
Then, for any $\vc{x} \in \C^N$ and $\vc{y} \in \C^m$ with $\| \vc{Ax} - \vc{y} \|_2 \leq \varrho$, any solution $\vc{x}^\#$ of
\eqref{eqn:basis_pursuit} fulfills \eqref{eqn:approx}, \eqref{eqn:approx2} where the constants $C,D > 0$ only depend on $\delta_{2s}$.
\end{theorem}

For certain random matrices, e.g., those with independent standard Gaussian entries, the following choice of the number of measurements $m$ depending on the dimension $N$ and the sparsity level $s$ is sufficient (and necessary, see below) for a small restricted isometry constant with high probability \cite{Foucart2013}:
\begin{equation}\label{eqn:typical_scaling}
m \gtrsim s \log (eN/s).
\end{equation}
Apart from Gaussian measurements this dependence also occurs for many other (even structured) random matrices --- up to possibly additional logarithmic factors. 

\subsection{An optimal RIP result}
A main goal of this paper consists in proving the following result on the number of required samples $N_t$ guaranteeing a small restricted isometry constant of the scaled MIMO radar measurement matrix (see \eqref{eqn:columns}, \eqref{eqn:columns2} for the definition)
\[
 \widetilde{\vc{A}} = \frac{1}{\sqrt{N_T N_R N_t}} \vc{A} \in \C^{N_R N_t \times N} ,
\]
and, moreover, in showing that this result is optimal in a certain sense.
We assume that the transmit pulse vectors $\vc{s}_i$ are independent standard complex Gaussian random vectors (see Appendix \ref{sec:tools} for some information about complex Gaussian random variables and vectors).

\begin{theorem}\label{thm:rip_result}
Let the signals $\vc{s}_1 , \vc{s}_2 , \ldots , \vc{s}_{N_T}$ generating the matrix $\vc{A}$ via \eqref{eqn:columns} and \eqref{eqn:columns2} be independent standard complex Gaussian random vectors.
If, for $\delta , \varepsilon \in (0,1)$,
\begin{equation}\label{eqn:rip_scaling}
 N_t \gtrsim \frac{s}{\delta^2} \max \{ \log^2 (eN) \log^2 (s) , \log (1 / \varepsilon ) \} ,
\end{equation}
then the restricted isometry constant of the rescaled matrix $\widetilde{\vc{A}}$ satisfies $\delta_s ( \widetilde{\vc{A}} ) < \delta$ with probability at least $1 - \varepsilon$.
\end{theorem}
\begin{remark}\label{rem:RademacherRIP} The result as well as its proof extends to signals $\vc{s}_1 , \vc{s}_2 , \ldots , \vc{s}_{N_T}$ being independent Rademacher vectors (random vectors with independent entries taking the value $+1$ or $-1$ with equal probability) or independent
Steinhaus vectors (random vectors with independent entries that are uniformly distributed on the complex torus).
These generators may be of advantage for real radar systems where one prefers constant magnitude signals for optimal energy
consumption and in order to avoid that amplifiers possibly run out of their linear regime.  
\end{remark}

Recall that the matrix $\widetilde{\vc{A}}$ has $m = N_R N_t$ rows, i.e., $\widetilde{\vc{A}}$ represents $N_R N_t$ measurements.
In this sense the theorem may seem suboptimal considering that one typically would expect a scaling as in \eqref{eqn:typical_scaling}.
However, using general theoretical results on lower bounds for the necessary number of measurements, we argue below that our RIP result above is in fact optimal (up to possibly logarithmic factors).

\subsubsection*{Instance optimality}
The concept of \textit{instance optimality} \cite{codade09,Foucart2013} allows to derive 
general lower bounds on the number of measurements $m$ required for stable sparse reconstruction via any algorithm 
and any possible measurement matrix in $\C^{m \times N}$.
To be precise, we say that a pair $(\boldsymbol{A},\varDelta)$, where $\varDelta \colon \C^m \to \C^N$ is a reconstruction map
and $\boldsymbol{A} \in \C^{m \times N}$,
is \textit{$\ell_1$-instance optimal of order $s$ with constant $C>0$} if for all $\vc{x} \in \C^N$ it holds
\[
\| \vc{x} - \varDelta (\boldsymbol{A} \vc{x}) \|_1 \leq C \inf_{\text{$s$-sparse $\vc{x}'$}} \| \vc{x}' - \vc{x} \|_1 .
\]
The following result states the announced lower bound.

\begin{theorem}[\mbox{\cite[Thm. 11.6]{Foucart2013}}]
If a pair of measurement matrix $\boldsymbol{A} \in \C^{m \times N}$ and reconstruction map $\varDelta \colon \C^m \to \C^N$ is $\ell_1$-instance optimal of order $s$ with constant C, then
\begin{equation}\label{eqn:necessary_measurements}
m \geq C^\prime s \log ( e N / s)
\end{equation}
for some constant $C^\prime$ depending only on $C$.
\end{theorem}

In case of the MIMO radar measurement matrix ${\vc{A}}$, this result can be applied to a certain submatrix of ${\vc{A}}$.
To this end, we define $\GG_{\beta^*} := \{ (\beta,\tau,f) \in \GG \midcol \beta = \beta^* \}$, for an arbitrary but fixed angle parameter $\beta^*$.
For a target scene $\vc{x}$ with support in $\GG_{\beta^*}$ one finds (recalling \eqref{eqn:columns}) that the product ${\vc{A}} \vc{x}$ can be expressed as outer product $\vc{w} \otimes {\vc{B}} \vc{x}_{\GG_{\beta^*}}$, where ${\vc{B}}$ is an $N_t \times N_t^2$ matrix with columns
\[
 \vc{B}_{(\beta^* ,\tau,f)} = \sum_{i=1}^{N_T} \ee{d_T \beta^* \Delta_\beta (i-1)} \vc{M}_{f} \vc{T}_{\tau} \vc{s}_i , \qquad (\beta^* ,\tau , f) \in \GG_{\beta^*} ,
\]
and $\vc{w}$ is a $N_R$-dimensional vector, entrywise defined as
\[
 [\vc{w}]_j = \ee{d_R \beta^* \Delta_\beta (j-1)} , \qquad j \in [N_R],
\]
and the vector $\vc{x}_{\GG_{\beta^*}}$ is obtained from $\vc{x}$ by deleting all entries which do not belong to the set $\GG_{\beta^*}$.
This implies that $\| \widetilde{\vc{A}} \vc{x} \|_2^2 = \| \widetilde{\vc{B}} \vc{x}_{\GG_{\beta^*}} \|_2^2$,
where $\widetilde{\vc{B}}$ stands for the scaled matrix $\frac{1}{\sqrt{N_T N_t}} \vc{B}$.
Consequently, the restricted isometry constant $\delta_{2s}(\widetilde{\vc{B}})$ is bounded by $\delta_{2s}(\widetilde{\vc{A}})$.
We conclude that $\delta_{2s}(\widetilde{\vc{A}}) < 1/\sqrt{2}$ implies $\delta_{2s}(\widetilde{\vc{B}}) < 1/\sqrt{2}$ so that
by Theorem~\ref{thm:rip_reconstr_result} the matrix $\widetilde{\vc{B}}$ in combination with the reconstruction map corresponding to basis pursuit \eqref{eqn:basis_pursuit} with $\rho=0$ is $\ell_1$-instance optimal. Since $\widetilde{\vc{B}}$ has $N_t$ rows and $N_t^2$ columns, \eqref{eqn:necessary_measurements} reads as $N_t \gtrsim s \log (eN_t^2 / s)$.
Therefore, we can formulate the following result. (Essentially, this result could have similarly been derived from Corollary~10.8 in \cite{Foucart2013}.)


\begin{theorem}
If the restricted isometry constant of the scaled MIMO radar measurement matrix \mbox{$\widetilde{\vc{A}} \in \C^{N_R N_t \times N}$} considered in Theorem~\ref{thm:rip_result} satisfies $\delta_{2s} < 1/\sqrt{2}$, then it holds
\begin{equation}\label{eqn:lower_bound}
 N_t \gtrsim s \log (eN_t^2 / s) .
\end{equation}
\end{theorem}

This shows that --- up to additional logarithmic factors --- Theorem \ref{thm:rip_result} on the restricted isometry constant for the MIMO radar measurement matrix is indeed optimal.
Analogously one can also argue that $\ell_1$-instance optimal uniform recovery in general --- which is not necessarily based on the RIP --- is not possible with less measurements than in \eqref{eqn:rip_scaling}.
Indeed, by assuming that there exists a corresponding reconstruction map for the matrix $\widetilde{\vc{A}}$ one can use this map for constructing an $\ell_1$-instance optimal reconstruction map for any submatrix in $\widetilde{\vc{B}}$ as defined above, which again yields \eqref{eqn:lower_bound}.

\subsection{Nonuniform recovery for random support sets}\label{sec:nonuniform_results}
As stated by Theorem \ref{thm:rip_result}, for a \textit{fixed} draw $\vc{A}$ of the measurement matrix the RIP is fulfilled with high probability and, hence (due to Theorem \ref{thm:rip_reconstr_result}), \textit{all} sparse vectors $\vc{x}$ can be reconstructed by considering measurements $\vc{A}\vc{x}$.
Since in this scenario $\vc{A}$ is fixed and reconstruction is possible for all sparse vectors $\vc{x}$, we speak of a \textit{uniform} recovery result.
A \textit{nonuniform} recovery result, on the other hand, only states that a given sparse vector $\vc{x}$ can be reconstructed from measurements $\vc{A}\vc{x}$ with high probability on the draw of the matrix $\vc{A}$, i.e., no assertion on the reconstruction of all sparse vectors using a single draw of the matrix $\vc{A}$ is made.

The results on the MIMO radar measurements by Strohmer \& Friedlander \cite{Strohmer2012} imply that, at least on average, \textit{nonuniform} reconstruction (for random support sets) succeeds with less measurements.
In \cite{Strohmer2012} they derive results for nonuniform reconstruction via the \textit{debiased LASSO} reconstruction scheme.
In case of the \textit{standard} LASSO, an approximation $\vc{x}^\#$ of $\vc{x}$ is obtained by solving the minimization problem
\begin{equation}\label{eqn:lasso}
\min_{\vc{z}} \frac{1}{2} \| \vc{Az} - \vc{y} \|_2^2 + \lambda \|\vc{z}\|_1 ,
\end{equation}
where $\lambda > 0$ has to be chosen appropriately (in accordance with the noise level).
In case that the support $S = \supp ( \vc{x} )$ is recovered correctly, i.e., $\supp ( \vc{x}^\# ) = S$, one might add a further step consisting of solving the reduced least squares problem 
\begin{equation}\label{eqn:debiasing}
 \min_{\vc{z}} \| \vc{A}_S \vc{z} - \vc{y}  \|_2^2 ,
\end{equation}
where $\vc{A}_S$ contains only the columns of $\vc{A}$ corresponding to the support set $S$ and the minimization is over all $\vc{z} \in \C^{|S|}$.
This second {\lq\lq debiasing\rq\rq} step leads in general to an improvement of the reconstructed coefficients.
The LASSO approach is closely related to basis pursuit denoising \eqref{eqn:basis_pursuit} but due to the unconstrained optimization often favorable in practice \cite{Foucart2013}.
\mbox{Strohmer \& Friedlander} assume a random target model, assuming that for a fixed sparsity parameter every support set occurs with the same probability and, moreover, the phases of the nonzero entries of the targets are independent and uniformly distributed in $[0,2\pi )$.
This model has been studied in \cite{capl09} in a general context and is called the \textit{generic $s$-sparse target model} therein; see also \cite[\mbox{Chap.~14}]{Foucart2013}.
Under additional technical assumptions, it has been shown in \cite{Strohmer2012}\footnote{Note that there is a typo in \cite[\mbox{Thm. 5}]{Strohmer2012} where an additional factor $N_t$ on the left-hand side of \eqref{eqn:strohmer_scaling} shows up.} that if $\vc{x}$ is an $s$-sparse target scene, then
\begin{equation}\label{eqn:strohmer_scaling}
 N_R N_t \gtrsim s \log (N) ,
\end{equation}
measurements are sufficient to guarantee that, with high probability, each solution $\vc{x}^\#$ to the LASSO has in particular the exact same support set as the target vector $\vc{x}$, thus perfectly recovering the parameters of the targets.

\subsection{Nonuniform recovery for deterministic support sets}\label{sec:correlation}
The results by Strohmer \& Friedlander imply that there must exist support sets which can be reconstructed using less measurements (than in our RIP result, see Theorem \ref{thm:rip_result}).
As it turns out, certain angle classes are not correlated among each other and, hence, support sets where the mass is uniformly distributed among these angle classes are favorable for reconstruction.

\subsubsection*{Balanced support sets}
Recall the definition of the columns $\vc{A}_\varTheta$, $\varTheta \in \GG$, of the MIMO radar measurement matrix $\vc{A}$ from \eqref{eqn:columns}, \eqref{eqn:columns2}.
For the inner product of two such columns $\vc{A}_{\varTheta}, \vc{A}_{\varTheta'}$ it holds
\begin{equation}\label{eqn:correlation_formula}
 \langle \vc{A}_{\varTheta} , \vc{A}_{\varTheta^\prime} \rangle = \sum_{k=1}^{N_R} \ee{d_R (\beta^\prime -\beta) \Delta_\beta (k-1)} \left\langle \sum_{i=1}^{N_T} \ee{d_T \beta \Delta_\beta (i-1)} \vc{M}_{f} \vc{T}_{\tau} \vc{s}_i , \sum_{i=1}^{N_T} \ee{d_T \beta^\prime \Delta_\beta (i-1)} \vc{M}_{f^\prime} \vc{T}_{\tau^\prime} \vc{s}_i \right\rangle .
\end{equation}
Due to the fact that $d_R = {N_T}/{2}$ and $\Delta_\beta = 2/N_T N_R$, the first factor can be calculated as
\[
 \sum_{k=1}^{N_R} \ee{d_R (\beta^\prime -\beta ) \Delta_\beta (k-1)} = \delta^{\sim_{N_R}}_{\beta,\beta'} N_R, \qquad \text{where } \quad \delta^{\sim_{N_R}}_{\beta,\beta'} :=
 \begin{cases}
 1, & \text{if $\beta^\prime - \beta \in N_R \Z$,} \\
 0, & \text{else.}
 \end{cases}
\]
This means that for two columns $\vc{A}_{(\beta,\tau,f)}$, $\vc{A}_{(\beta^\prime,\tau^\prime,f^\prime)}$ to be correlated it is necessary that $\beta^\prime - \beta \in N_R \Z$.
The latter condition induces an equivalence relation on the set of possible angle parameters.
We say that two angle parameters $\beta , \beta^\prime$ are equivalent if and only if $\beta^\prime - \beta \in N_R \Z$, i.e.,
\[
 \beta^\prime \sim \beta \iff \beta^\prime - \beta \in N_R \Z ,
\]
which, recalling from \eqref{eqn:grid} that $\beta$ is taken from $\{ 1 , 2 , \ldots , N_T N_R \}$, leaves us with $N_R$ angle classes $\betacl$, each containing $N_T$ different angle parameters.
The crucial fact is that, since the columns from different angle classes are uncorrelated, the matrices $\widetilde{\vc{A}}_S^* \widetilde{\vc{A}}_S$, $S \subset \GG$, are block diagonal.
It holds in particular,
\begin{equation}\label{eqn:matrix_AS_block_diagonal}
 \opnorm{\widetilde{\vc{A}}_S^* \widetilde{\vc{A}}_S - \ID} =  \max_{[\beta]} \opnorm{\widetilde{\vc{A}}_{S_{[\beta]}}^* \widetilde{\vc{A}}_{S_{[\beta]}} - \ID},
\end{equation}
where the subsets $S_{\betacl} \subseteq S$ are defined as
\begin{equation}\label{eqn:equivalent_indices}
 S_{[\beta]} := \{ \varTheta^\prime = ( \beta^\prime , \tau^\prime , f^\prime ) \in S \midcol \beta^\prime \sim \beta \} .
\end{equation}
Considering \eqref{eqn:matrix_AS_block_diagonal} one may expect that a uniform distribution of any given support set $S$ over all angle classes would make it most likely that the matrix $\widetilde{\vc{A}}_S^* \widetilde{\vc{A}}_S$ is well conditioned (since in this case the maximal dimension of the submatrices $\widetilde{\vc{A}}_{S_{[\beta]}}^* \widetilde{\vc{A}}_{S_{[\beta]}}$ becomes minimal).
This motivates to introduce a parameter $\eta$ measuring how evenly distributed over the angle classes such a support set is.

\begin{definition}\label{def:balanced}
The balancedness parameter of a support set $S \subset \GG$ is the smallest number $\eta$ such for all angle classes $[\beta]$ it holds
\begin{equation*}
 |S_{[\beta]}| \leq \eta \frac{|S|}{N_R} .
\end{equation*}
We say that $S$ is $\eta$-balanced, if the balancedness parameter of $S$ equals $\eta$.
\end{definition}

According to this definition the balancedness parameter can expressed as
\[
 \eta = \max_{[\beta]} \frac{N_R |S_{[\beta]}|}{|S|} .
\]
Clearly, the balancedness parameter takes values in $[1, N_R]$, where $\eta = 1$ means that the $|S|$ targets are perfectly balanced with respect to the $N_R$ angle classes (with $|S_{[\beta]}| \sim |S| /N_R$), and, on the other hand, $\eta=N_R$ corresponds to the case where all targets are located in one particular angle class.

\subsubsection*{Nonuniform recovery results}
Next we present our nonuniform recovery results for $s$-sparse target scenes with \mbox{$\eta$-balanced} support sets (see Definition \ref{def:balanced} above) and, thereby, reveal an essentially linear dependence of the number of measurements $m = N_R N_t$ required for reconstruction on the number of targets $s$.
Note, how the parameter $\eta$ enters the corresponding condition \eqref{eqn:nonuniform_scaling} below.
In the worst possible case where $\eta=N_R$ we obtain essentially the same scaling as for the RIP estimate in \eqref{eqn:rip_scaling}, while for the best possible case $\eta = 1$ we obtain the announced linear scaling.

For the following we recall that a Steinhaus sequence is a sequence of independent random variables that are uniformly distributed on the complex unit circle $\{ z \in \C \midcol |z|=1 \}$.
Due to the fact that the grid $\GG$ (see \eqref{eqn:grid}) is $N=N_T N_R N_t^2$-dimensional, a target scene $\vc{x}$ can be regarded as a vector in $\C^N$.

\begin{theorem}\label{thm:nonuniform_result}
Let $\vc{x} \in \C^{N}$ be an $s$-sparse target scene with $\eta$-balanced support set and assume that measurements $\vc{y} = \vc{Ax} + \vc{n} \in \C^{N_R N_t}$ are given, where the signals $\vc{s}_1 , \vc{s}_2 , \ldots , \vc{s}_{N_T}$ generating the columns of $\vc{A}$ (according to \eqref{eqn:columns}, \eqref{eqn:columns2}) are independent standard complex Gaussian random vectors, and the noise vector $\vc{n}$ is a mean-zero complex Gaussian random vector with variance $\sigma^2$.
If
\begin{equation}\label{eqn:nonuniform_scaling}
 N_R N_t \gtrsim \eta s \log^3 ( N / \varepsilon ) ,
\end{equation}
and the signs of the nonzero entries of $\vc{x}$ form a Steinhaus sequence while the magnitudes satisfy
\begin{equation}\label{eqn:thres}
 \min_{\varTheta \in \supp ( \vc{x} )} |\vc{x}_\varTheta| > \frac{8 \sigma}{\sqrt{N_T N_R N_t}} \sqrt{2 \log (N)} ,
\end{equation}
then, with probability at least $1- 7 \max \{ \varepsilon , N^{-2} \}$,
\begin{enumerate}
 \item[(a)] the solution $\vc{x}^\#$ to the LASSO \eqref{eqn:lasso} with $\lambda = 2\sigma \sqrt{2 N_T N_R N_t \log (N)}$ satisfies
 \[
  \supp ( \vc{x}^\# ) = \supp (\vc{x}) ,
 \]
 \item[(b)]
 the solution $\vc{z}^\#$ to the least squares problem \eqref{eqn:debiasing} with $S = \supp ( \vc{x}^\# )$ satisfies
 \begin{equation}\label{eqn:coeff_estimate}
  \| \vc{z}^\# - \vc{x}_S \|_2 \leq \frac{2 \sigma \sqrt{s}}{\sqrt{N_T N_R N_t}} \sqrt{2 \log (N)} .
 \end{equation}
\end{enumerate}
\end{theorem}

\goodbreak
Note that, in view of the scaling of $s$ in \eqref{eqn:nonuniform_scaling}, the approximation estimate \eqref{eqn:coeff_estimate} actually implies that $\| \vc{z}^\# - \vc{x}_S \|_2 \lesssim \sigma / \sqrt{N_T}$, where $N_T$ is the number of transmitters.
Furthermore, we would like to point out that, for technical reasons, 
we did not consider deterministic signs of the nonzero entries of $\vc{x}$.
Nevertheless, we do believe that with similar ideas as in \cite{hürast14} such a further generalization would be possible (although technical).

The next theorem, our second main result on nonuniform recovery, makes an assertion about reconstruction via basis pursuit denoising.
Unlike the previous Theorem \ref{thm:nonuniform_result}, the following result deals with target scenes which are not exactly sparse but rather can be approximated well by sparse vectors.
Again, the balancedness parameter $\eta$ of the considered support sets is essential.

\begin{theorem}\label{thm:nonuniform_result_sparsity_defect}
 Let $\vc{x} \in \C^{N}$ be a target scene, let $S \subset \GG$ be an index set corresponding to $s$ largest (by magnitude) entries in $\vc{x}$, and let $S$ be $\eta$-balanced. Further, assume that the signs of the coefficients $\vc{x}_S$ form a Steinhaus sequence.
 Assume measurements \mbox{$\vc{y} = \vc{Ax} + \vc{n} \in \C^{N_R N_t}$} are given, where the signals $\vc{s}_1 , \vc{s}_2 , \ldots , \vc{s}_{N_T}$ generating the columns of the measurement matrix $\vc{A}$ \mbox{(according to \eqref{eqn:columns}, \eqref{eqn:columns2})} are independent standard Gaussian random vectors, and the noise vector $\vc{n}$ is known to satisfy $\|\vc{n}\|_2 \leq \varrho$.
 If
 \begin{equation}\label{eqn:nonuniform_scaling2}
  N_R N_t \gtrsim \eta s \log^3 ( N / \varepsilon ) ,
 \end{equation}
 then, with probability at least $1-\varepsilon$, the solution $\vc{x}^\#$ to the basis pursuit denoising program \eqref{eqn:basis_pursuit} satisfies
 \[
  \| \vc{x}^\# - \vc{x} \|_2 \leq C_1 \inf_{\text{$s$-sparse $\vc{x}'$}} \| \vc{x}' - \vc{x} \|_1 + C_2 \frac{\varrho\sqrt{s}}{\sqrt{N_T N_R N_t}} ,
 \]
where $C_1$ and $C_2$ are numerical constants.
\end{theorem}

\begin{remark} Like Theorem~\ref{thm:rip_result} concerning the RIP, Theorem~\ref{thm:nonuniform_result} and \ref{thm:nonuniform_result_sparsity_defect} extend to signals $\vc{s}_1 , \vc{s}_2 , \ldots , \vc{s}_{N_T}$
being independent Rademacher and Steinhaus sequences, compare also with Remark~\ref{rem:RademacherRIP}.
\end{remark}

\subsection{The Doppler-free scenario}\label{sec:doppler_free_case}
In \cite{Strohmer2012} also a Doppler-free scenario (for the case of slowly moving or stationary targets) has been analyzed.
We want to point out that this case is also covered by our analysis.
The measurement matrix $\vc{A}^\prime$ corresponding to the Doppler-free case can be obtained from $\vc{A}$ by deleting all columns $\vc{A}_{(\beta , \tau , f)}$ with nonzero Doppler shifts $f$.
Therefore, in view of \eqref{eqn:rip}, our RIP result (see Theorem \ref{thm:rip_result}) holds also automatically true for the Doppler-free case --- where one may even replace the factor $\log (N) = \log (N_T N_R N_t^2 )$ by $\log ( N_T N_R N_t )$.
The proof of our nonuniform recovery results, stated as Theorems \ref{thm:nonuniform_result} and \ref{thm:nonuniform_result_sparsity_defect}, is based on Proposition \ref{prop:opnorm_bound} providing estimates for the singular values of the matrices $\vc{A}_S$ where $S$ stands for an arbitrary (but balanced) support set.
These estimates, of course, hold in particular true when considering only matrices $\vc{A}_S$ with the property that $S$ does not contain any index $(\beta,\tau,f)$ with $f\neq 0$ and is balanced.
The reader might assure himself that all remaining arguments in Section \ref{sec:nonuniform_deterministic} hold equally true for the Doppler-free case and, thus, analog results on nonuniform recovery can be proven.

\subsection*{Outline}
In Section \ref{sec:rip} we prove the RIP result (Theorem \ref{thm:rip_result}) by reformulating the restricted isometry constant as the supremum of a chaos process and then applying a general result on such processes shown in \cite{krmera14,Dirksen2013}.
Section \ref{sec:nonuniform_deterministic} is devoted to the proofs of our nonuniform recovery result for balanced support sets (Theorems \ref{thm:nonuniform_result} and \ref{thm:nonuniform_result_sparsity_defect}).
The proof of Theorem~\ref{thm:nonuniform_result} is based on a general reconstruction result for the LASSO approach taken from \cite{capl09} which we introduce in Section \ref{sec:exact_recovery_via_lasso}.
The Sections~\ref{sec:conditions_lemmas}--\ref{sec:proof_nonuniform_result_b} are devoted to the verification of the assumption needed for applying the reconstruction result.
To this end a probabilistic analysis of the extremal singular values of the matrices $\widetilde{\vc{A}}_S$, which we state as Proposition \ref{prop:opnorm_bound}, is essential.
Since the proof of Proposition \ref{prop:opnorm_bound} is rather extensive it appears in an extra section, namely Section \ref{apdx:proof_prop}.
In Section \ref{sec:sparsity_defect_result} we prove Theorem \ref{thm:nonuniform_result_sparsity_defect} by applying a more general recovery result for the basis pursuit denoising approach, taken from \cite{Foucart2013}.
Also in this case the verification of the assumptions depends crucially on the analysis of the singular values of $\widetilde{\vc{A}}_S$ provided by Proposition \ref{prop:opnorm_bound}.
Finally, Section \ref{sec:numerics} is devoted to some numerical experiments supporting our theoretical results on the influence of the balancedness parameter $\eta$ (of the considered support sets) on the recovery performance.

In the Appendix we list basic calculations and technical proofs which for the sake of a clearer presentation do not appear in the main part.
Furthermore, some tools from probability theory and a short introduction of standard complex Gaussian random variables and vectors can be found here.

\subsection*{Notation}
For a complex number $z$ we write $\overline{z}$ for the complex conjugate and, for nonzero $z$, we set $\text{sgn} (z) := z / |z|$.
Given a vector $\vc{x} = (x_1 , x_2 , \ldots )^T$ with complex entries, we write $\| \vc{x} \|_p = ( \sum_k |x_k|^p )^{1/p}$, $1\leq p < \infty$, to denote the usual $\ell_p$-norm of $\vc{x}$, and \mbox{$\| \vc{x} \|_\infty = \max_{k} | x_k |$}.
Moreover, $\text{supp} ( \vc{x} ) = \{ k \midcol x_k \neq 0 \}$ denotes the support set of $\vc{x}$.
Occasionally we also write $[\vc{x}]_k$ to denote the $k$th entry of the vector $\vc{x}$.
We write $\ID$ for the identity matrix (with appropriate dimensions becoming clear from the context).
Given a complex valued matrix $\vc{A}$, we write $[\vc{A}]_{k,\ell}$ to denote the $k$th entry of the $\ell$th row of $\vc{A}$.
The Hermitian transpose will be denoted by $\vc{A}^*$.
The spectral norm is denoted by $\opnorm{\vc{A}} = \max_{\|\vc{x}\|_2 =1} \|\vc{A}\vc{x}\|_2$. The Frobenius norm of a matrix $\vc{A}$ is given by $\|\vc{A}\|_{F} = \sqrt{\sum_{k,\ell} [\vc{A}]_{k,\ell}^2}$.
Given the sequence of singular values $\sigma (\vc{A}) = ( \sigma_1 (\vc{A}) , \sigma_2 (\vc{A}) , \ldots )$ of a matrix $\vc{A}$, the Schatten $p$-norm, for $1\leq p\leq \infty$, is given by $\|\vc{A}\|_{S_p} := \| \sigma( \vc{A} ) \|_p$. Note the special
cases $\|\vc{A}\|_{S_2} = \|\vc{A}\|_F$ and $\|\vc{A}\|_{S_\infty} = \opnorm{\vc{A}}$.
For a set $\AA$ of matrices, the parameters $d_{2\to 2} ( \AA ) = \sup_{A \in \AA} \opnorm{\vc{A}}$, $d_F ( \AA ) = \sup_{A \in \AA} \|\vc{A}\|_{F}$ denote the diameters of $\AA$ with respect to the spectral norm and the Frobenius norm, respectively.
Given a number $L\in \N$, we write $[L]$ to denote the set $\{1, 2, \ldots, L\}$.
We will also make use of the symbol $\delta_{k,\ell}$ or, more generally, $\delta^{\sim_{L}}_{k,\ell}$, for $k,\ell \in \Z$ and $L\in \N$.
Here $\delta_{k,\ell}$ is the usual Kronecker delta which is equal to $1$ if $k=\ell$ and $0$, otherwise.
The symbol $\delta^{\sim_{L}}_{k,\ell}$, on the other hand, represents equality up to multiples of $L$, i.e.,
\begin{equation}\label{eqn:general_kronecker}
 \delta^{\sim_{L}}_{k,\ell} =
 \begin{cases}
  1  &  \text{if $k - \ell \in L \Z$,}  \\
  0  &  \text{otherwise.}
 \end{cases}
\end{equation}
We write $A \lesssim B$ if there is a constant $c_1$ with $A \leq c_1 B$, where $A,B$ may depend on further parameters.

\goodbreak
\section{Uniform recovery via the RIP}\label{sec:rip}
In this section we prove Theorem \ref{thm:rip_result} concerning the RIP of the MIMO radar measurement matrix.
In the following, $\vc{s}$ denotes the vector containing all signals $\vc{s}_i$, i.e.,
\[
 \vc{s} = ( \vc{s}_1^T , \vc{s}_2^T , \ldots , \vc{s}_{N_T}^T )^T \in \C^{N_T N_t} .
\]

\subsection{Reformulation as a chaos process}

Our proof is based on \cite{krmera14} where bounds for certain chaos processes have been developed.
The key is to reformulate the radar measurements $\vc{x} \mapsto \widetilde{\vc{A}} \vc{x}$ as a mapping, $\vc{x} \mapsto \vxTilde \vc{s}$, where the matrix $\vxTilde$, depending on $\vc{x}$, is given by
\begin{equation}\label{eqn:vx}
\vxTilde := \frac{1}{\sqrt{N_T N_R N_t}} \sum_{\varTheta \in \GG} x_{\varTheta} \x{\varTheta} ,
\end{equation}
and where $\x{\varTheta}$ is a $N_R \times N_T$ block matrix consisting of the $N_t \times N_t$ blocks $\x{\varTheta}^{[i,j]}$ with
\begin{equation}\label{eqn:XX}
\x{\varTheta}^{[i,j]} := \ee{d_R \beta \Delta_\beta (i-1)} \ee{d_T \beta \Delta_\beta (j-1)} \vc{M}_{f} \vc{T}_{\tau} .
\end{equation}
Due to the fact that the signal vector $\vc{s}$ is isotropic, we have for each $\vc{x}$ that $\E \| \vxTilde \vc{s} \|_2^2 = \| \vxTilde \|_{{F}}^2$, and, due to the fact that the matrices $\frac{1}{\sqrt{N_T N_R N_t}} \x{\varTheta}$, $\varTheta \in \GG$, are orthogonal (see \mbox {Appendix \ref{apdx:section2}}), $\| \vxTilde \|_{{F}}^2 = \|\vc{x}\|_2^2$.
This enables us to express the restricted isometry constant $\delta_s$ as supremum of a chaos process of the form
\begin{equation}\label{eqn:chaos_reform}
 \delta_s = \sup_{\genfrac{}{}{0pt}{}{\text{$s$-sparse $\vc{x}$}}{\| \vc{x} \|_2 \leq 1}} \big| \|\widetilde{\vc{A}} \vc{x}\|_2^2 - \|\vc{x}\|_2^2 \big| = \sup_{\widetilde{\vc{V}} \in \AA} \big| \| \widetilde{\vc{V}} \vc{s} \|_2^2 - \E \|\widetilde{\vc{V}} \vc{s}\|_2^2 \big| ,
\end{equation}
where $\AA$ denotes the set of matrices given by
\begin{equation}\label{eqn:AA_set}
 \AA := \{ \vxTilde \midcol \text{$\vc{x}$ is $s$-sparse, $\|\vc{x}\|_2 \leq 1$} \} .
\end{equation}
Dudley's entropy integral
\begin{equation}\label{eqn:dudley}
\DD := \int_0^{d_{2\to 2} ( \AA )} \sqrt{\log \NN( \AA , \opnorm{\cdot} , u )} \, du
\end{equation}
for the set $\AA$, where $\NN( \AA , \opnorm{\cdot} , u)$ denotes the covering numbers, i.e., the minimal number of 
$\opnorm{\cdot}$-balls of radius $u$ required to cover $\AA$, provides sufficient information about the complexity of the set $\AA$ for us.
Estimates for the parameter $\DD$ will in fact lead to good bounds for the restricted isometry constant in \eqref{eqn:chaos_reform}.
The following theorem --- which is a direct implication of \cite[Thm.~3.1]{krmera14} and a slight improvement due to Dirksen \cite{Dirksen2013} --- is our crucial tool.
Recall 
that a random vector $\vc{X}$ is $L$-subgaussian if it is isotropic and $\P ( |\langle \vc{X} , \vc{\xi} \rangle | \geq t) \leq 2 \exp (-t^2 / 2 L^2)$ for every $\vc{\xi}$ with $\|\vc{\xi}\|_2 = 1$ and any $t >0$.

\begin{theorem}\label{thm:chaos_tool}
Let $\AA$ be a symmetric set of matrices, $-\AA = \AA$, and let $\vc{\xi}$ be a random vector whose entries are independent, mean-zero, variance $1$, and $L$-subgaussian random variables. Set
\[
 E = \DD \big( \DD + d_F (\AA) \big) ,
\]
\[
 V = d_{2\to 2}(\AA) d_F (\AA) ,
\]
\[
 U = d_{2\to 2}^2 (\AA).
\]
Then, for $t > 0$,
\[
\P \big( {\sup_{\vc{A} \in \AA} \big| \| \vc{A \xi} \|^2_2 -\E \|\vc{A\xi}\|^2_2 \big|} \geq c_1 E + t \big)
  \leq 2 \exp (-c_2 \min \{ t^2 / V^2 , t / U \}),
\]
where the constants $c_1, c_2$ depend only on $L$.
\end{theorem}

Note that the original result  \cite[Thm. 3.1]{krmera14} is formulated in terms of Talagrand's chaining functional $\gamma_2$, whereas Theorem \ref{thm:chaos_tool} from above is formulated in terms of Dudley's entropy integral $\DD$, a common upper estimate for the $\gamma_2$-functional, which is sufficient for our purposes and more convenient.
For more details on $\gamma_2$-functionals we refer to \cite{Talagrand2005,ta14-1}.
In the original version of Theorem~\ref{thm:chaos_tool} in \cite{krmera14} the quantity $V$ was defined as 
$V=d_{2\to 2}(\AA) (\DD + d_F (\AA))$.
Dirksen \cite{Dirksen2013} achieved a slight improvement by showing that the summand $\DD$ can be omitted.

\subsection{Proof of Theorem \ref{thm:rip_result}}
In order to apply the above theorem we have to bound $d_{F}(\AA)$, $d_{2 \to 2}(\AA)$, and $\DD$.
Estimates for the quantities $d_{2 \to 2}(\AA)$ and $d_{F}(\AA)$ can be derived in a straightforward manner.
We have already seen that $\| \vxTilde \|_F = \| \vc{x} \|_2$ and, thus, $d_{{F}}(\AA) = 1$.
In order to provide a bound for $d_{2 \to 2}(\AA)$, we estimate the norms $\opnorm{\x{\varTheta}}$ of the matrices from the definition in \eqref{eqn:vx}.
The $(j^\prime ,j)$th $N_t \times N_t$ block of the product $\x{\varTheta}^* \x{\varTheta}$ can be calculated (cf. Appendix \ref{apdx:section2}, equation \eqref{eqn:block_product_XX}) as
\[
 [ \x{\varTheta}^* \x{\varTheta} ]^{[j^\prime , j]} = N_R \nee{d_T \beta \Delta_\beta (j^\prime-1)} \ee{d_T \beta \Delta_\beta (j - 1)} \ID .
\]
Thus, by defining a vector $\vc{v} \in \C^{N_T}$ entrywise as
\[
 [\vc{v}]_{j} = \ee{d_T \beta \Delta_\beta (j - 1)} , \quad j \in [N_T] ,
\]
the corresponding operator norm can be calculated as
\[
 \opnorm{\x{\varTheta}}^2 = \opnorm{\x{\varTheta}^* \x{\varTheta}} = N_R \opnorm{ \vc{v} \vc{v}^* } = N_R \|\vc{v}\|_2^2 = N_R N_T .
\]
Since for each $s$-sparse vector $\vc{x}$ we have $\| \vc{x} \|_1 \leq \sqrt{s} \|\vc{x}\|_2$, we obtain
\begin{equation}\label{eqn:opnorm_by_euclid}
 \opnorm{ \vxTilde} \leq \frac{1}{\sqrt{N_T N_R N_t}} \sum_{\varTheta \in \GG} |x_{\varTheta}| \, \opnorm{\x{\varTheta}} \leq \frac{1}{\sqrt{N_t}} \| \vc{x} \|_1 \leq \sqrt{\frac{s}{N_t}} \|\vc{x}\|_2 .
\end{equation}
For each $\vxTilde \in \AA$ we have, by definition, that $\|\vc{x}\|_2 \leq 1$ which means that 
\begin{equation}\label{eqn:d2to2_bound}
d_{{F}} (\AA) = 1 , \qquad \text{and} \qquad d_{2 \to 2}(\AA) \leq \sqrt{\frac{s}{N_t}} .
\end{equation}

In order to estimate the entropy integral $\DD$ in \eqref{eqn:dudley}, it is in particular necessary to provide good bounds for the appearing covering numbers $\NN (\AA , \opnorm{\cdot} , u)$.
In view of \eqref{eqn:d2to2_bound} it is sufficient to consider $0 < u \leq \sqrt{s/N_t}$.
The proof of Lemma \ref{lem:cov_number_bounds} can be found in Appendix \ref{apdx:proof_cov_numbers}.

\begin{lemma}\label{lem:cov_number_bounds}
For $0 < u \leq \sqrt{s/N_t}$, it holds
\begin{equation}\label{eqn:covering_numbers_bound}
 \log \NN(\AA, \opnorm{\cdot}, u) \lesssim s \min \left\{ \log \bigg( \frac{N}{u^2 N_t} \bigg) , \frac{1}{u^2 N_t} \log^2 (N) \right\}.
\end{equation}
\end{lemma}

Now we are able to estimate Dudley's integral as
\begin{equation}\label{eqn:asymptotic_DD}
 \DD \lesssim \sqrt{\frac{s}{N_t}} \log (eN) \log (s).
\end{equation}
To this end, we split the integral into two pieces,
\begin{align*}
 \DD
 &= \int_0^{d_{2\to 2} (\AA)} \sqrt{ \log \NN (\AA , \opnorm{\cdot} , u)} ~ du  \\
 &\leq \int_0^{\frac{1}{\sqrt{N_t}}} \sqrt{ \log \NN (\AA , \opnorm{\cdot} , u) } ~ du
  + \int_{\frac{1}{\sqrt{N_t}}}^{\frac{\sqrt{s}}{\sqrt{N_t}}} \sqrt{ \log \NN (\AA , \opnorm{\cdot} , u) } ~ du =: \II_1 + \II_2 .
\end{align*}
We estimate the first integral by using the first bound of Lemma \ref{lem:cov_number_bounds} to obtain
\begin{equation*}
 \II_1 \lesssim \int_0^{\frac{1}{\sqrt{N_t}}} \sqrt{s \log \bigg( \frac{N}{u^2 N_t} \bigg)} ~ du \leq \sqrt{s} \sqrt{ \int_0^{\frac{1}{\sqrt{N_t}}} 1 ~ du \int_0^{\frac{1}{\sqrt{N_t}}} \log \bigg( \frac{N}{u^2 N_t} \bigg) ~ du } ,
\end{equation*}
where, in the second step, we used the Cauchy-Schwarz inequality.
The latter integral can easily be calculated as
\begin{align*}
 \int_0^{\frac{1}{\sqrt{N_t}}} \log \bigg( \frac{N}{u^2 N_t} \bigg) ~ du
  &= \int_0^{\frac{1}{\sqrt{N_t}}} \log (N/N_t) - \log (u^2) ~ du  
  = \frac{1}{\sqrt{N_t}} \log ( N/N_t ) - \bigg[ u \big(\log (u^2) - 2\big) \bigg]_0^{1 / \sqrt{N_t}}  \\
  &= \frac{1}{\sqrt{N_t}} \big( \log (N) + 2 \big) ,
\end{align*}
which, together with the estimate from above, yields
\[
 \II_1 \lesssim \sqrt{\frac{s}{N_t}} \sqrt{\log(N)+2} .
\]
For the second integral $\II_2$ we use the second bound provided by Lemma \ref{lem:cov_number_bounds} and obtain
\begin{equation*}
 \II_2 \lesssim \int_{\frac{1}{\sqrt{N_t}}}^{\frac{\sqrt{s}}{\sqrt{N_t}}} \sqrt{ \frac{s}{u^2 N_t} \log^2 (N)} ~ du = \sqrt{\frac{s}{N_t}} \log (N) \int_{\frac{1}{\sqrt{N_t}}}^{\frac{\sqrt{s}}{\sqrt{N_t}}} u^{-1} ~ du = \sqrt{\frac{s}{N_t}} \log (N)  \log (\sqrt{s}) .
\end{equation*}
A combination of the estimates for $\II_1$ and $\II_2$ from above yields \eqref{eqn:asymptotic_DD}.

Now, having a bound for the entropy integral $\DD$, we are well equipped to apply Theorem \ref{thm:chaos_tool} and, thereby, prove Theorem \ref{thm:rip_result}.
To this end let $c > 0$ denote a constant to be chosen sufficiently large and
\begin{equation}\label{eqn:condition_on_Nt}
N_t \geq c \delta^{-2} s \max \left\{ \log^2 (eN) \log^2 (s) , \log( 1 / \varepsilon ) \right\} .
\end{equation}
Due to \eqref{eqn:asymptotic_DD} it follows that, for some constant $C > 0$,
\[
\DD \leq C \sqrt{\frac{s}{N_t}} \log (eN) \log (s) .
\]
These two bounds imply that $\DD \leq {C\delta}/{\sqrt{c}}$ and, therefore, since $d_{{F}} (\AA) = 1$, we obtain for the quantity $E$ from Theorem \ref{thm:chaos_tool} (if $c$ is chosen sufficiently large),
\[
E = \DD ( \DD + 1 ) \leq (C \delta )^2 / c + C \delta / \sqrt{c} < \frac{\delta}{2 c_1} ,
\]
where $c_1$ is the constant from Theorem \ref{thm:chaos_tool}.
A direct application of Theorem \ref{thm:chaos_tool}, again using $d_{{F}} (\AA) = 1$, yields
\[
 \P ( \delta_s \geq \delta ) \leq \P ( \delta_s \geq c_1 E + {\delta}/{2}) \leq 2 \exp \left( - \frac{c_2 \delta^2}{4 d_{2 \to 2}^2 (\AA)} \right) \leq 2 \exp \left( - \frac{c_2 \delta^2 N_t}{4s} \right),
\]
where for the last inequality we used the estimate for $d_{2 \to 2} (\AA)$ from \eqref{eqn:d2to2_bound}.
Due to \eqref{eqn:condition_on_Nt}, the latter term is bounded by $\varepsilon$, again provided that $c$ is sufficiently large, which finalizes the proof of Theorem \ref{thm:rip_result}.\qed

\section{Nonuniform recovery results}\label{sec:nonuniform_deterministic}
The proofs of our nonuniform recovery results (see Theorems \ref{thm:nonuniform_result} and \ref{thm:nonuniform_result_sparsity_defect}) rely on the fact that --- for sufficiently well balanced support sets $S\subset \GG$ --- the singular values of the matrix $\widetilde{\vc{A}}_S$ are close to one with high probability.
Along with further properties, this enables us to apply general recovery results for both the LASSO and basis pursuit denoising which we take from \cite{capl09} and \cite{Foucart2013}, respectively.
In the following analysis we consider a scaled version of the considered measurement process, namely
\[
 \widetilde{\vc{y}} = \widetilde{\vc{A}} \vc{x} + \widetilde{\vc{n}}, \qquad \widetilde{\vc{y}} = \frac{1}{\sqrt{N_T N_R N_t}} \vc{y}, \qquad \widetilde{\vc{A}} = \frac{1}{\sqrt{N_T N_R N_t}} \vc{A}, \qquad \widetilde{\vc{n}} = \frac{1}{\sqrt{N_T N_R N_t}} \vc{n} .
\]
In the situation of Theorem~\ref{thm:nonuniform_result},  $\widetilde{\vc{n}}$ is now a mean-zero complex Gaussian random vector with variance $\tilde{\sigma}^2 = \sigma^2 / {N_T N_R N_t}$.

\subsection{Exact support recovery via the (debiased) LASSO}\label{sec:exact_recovery_via_lasso}
In order to prove Theorem~\ref{thm:nonuniform_result} we use a general recovery result from \cite{capl09},
which provides the sufficient conditions $(C_1)$--$(C_5)$ below on fixed instances of the measurement matrix $\widetilde{\vc{A}}$, the target scene vector $\vc{x}$ and the noise vector $\widetilde{\vc{n}}$ that imply perfect reconstruction of the support set $\supp (\vc{x})$ from the measurements $\widetilde{\vc{y}}$ via the LASSO
\begin{equation}\label{eqn:scaled_lasso}
 \min_{\vc{x}} \frac{1}{2} \| \widetilde{\vc{A}} \vc{x} - \widetilde{\vc{y}} \|_2^2 + 2 \mu \| \vc{x} \|_1 , \qquad \mu = \tilde{\sigma} \sqrt{2 \log (N)}.
\end{equation}
Later (see Sections \ref{sec:proof_nonuniform_result_a}, \ref{sec:proof_nonuniform_result_b} below) we will show that these conditions are fulfilled with high probability.
In the following we write $\Pi_S$ to denote the projection onto the linear space spanned by the columns of $\widetilde{\vc{A}}_S$.

\goodbreak
\begin{enumerate}[(i)]
\item[$(C_1)$] The matrix $\widetilde{\vc{A}}_S^* \widetilde{\vc{A}}_S$ is invertible and obeys
      $\opnorm{(\widetilde{\vc{A}}_S^* \widetilde{\vc{A}}_S)^{-1}} \leq 2$,
\item[$(C_2)$] $\| \widetilde{\vc{A}}^*_{S^c} \widetilde{\vc{A}}_S (\widetilde{\vc{A}}_S^* \widetilde{\vc{A}}_S)^{-1} \sgn (\vc{x}_S)\|_\infty < 1/4$,
\item[$(C_3)$] $\|(\widetilde{\vc{A}}_S^* \widetilde{\vc{A}}_S)^{-1} \widetilde{\vc{A}}_S^* \widetilde{\vc{n}}\|_\infty \leq 2 \mu$,
\item[$(C_4)$] $\| \widetilde{\vc{A}}^*_{S^c} (\ID - \Pi_S) \widetilde{\vc{n}} \|_\infty \leq \sqrt{2} \mu$,
\item[$(C_5)$] $\| (\widetilde{\vc{A}}_S^* \widetilde{\vc{A}}_S)^{-1} \sgn (\vc{x}_S) \|_\infty \leq 3$.
\end{enumerate}

In \cite{capl09} it is shown, using Conditions $(C_1)$--$(C_5)$, that the difference vector $\vc{h} = \hat{\vc{x}} - \vc{x}$ between the solution $\hat{\vc{x}}$ to the LASSO and the original vector $\vc{x}$ fulfills $\text{supp} (\vc{h}) \subset \supp (\vc{x})$ so that $\text{supp}(\hat{\vc{x}}) = \text{supp}(\vc{x})$ and, even more, the explicit formula \eqref{eqn:formula_h} for $\vc{h}$ is given.

\begin{lemma}[cf. \mbox{\cite[Lem. 3.4]{capl09}}]\label{lem:candes}
Let $\widetilde{\vc{y}} = \widetilde{\vc{A}} \vc{x} + \widetilde{\vc{n}}$ be given, and suppose that Conditions $(C_1)$--$(C_5)$ hold true, with $S = \supp (\vc{x})$.
If
\begin{equation*}
 \min_{\varTheta \in S} | {x}_{\varTheta} | > 8 \mu ,
\end{equation*}
then the solution to the LASSO \eqref{eqn:scaled_lasso} is given by $\hat{\vc{x}} = \vc{x} + \vc{h}$, with
\begin{equation}\label{eqn:formula_h}
\begin{array}{lcl} \vc{h}_S   &=&   (\widetilde{\vc{A}}_S^* \widetilde{\vc{A}}_S)^{-1} \big( \widetilde{\vc{A}}_S^* \widetilde{\vc{n}} - 2 \mu \sgn (\vc{x}_S) \big),   \\   \vc{h}_{S^c}   &=&   \vc{0}. \end{array}
\end{equation}
\end{lemma}

\subsubsection{Estimates for the conditions \texorpdfstring{$\vc{(C_1)}$--$\vc{(C_5)}$}{(C1)--(C5)}}\label{sec:conditions_lemmas}
With Lemma \ref{lem:candes} at hand we have a basic pattern for the proof Theorem \ref{thm:nonuniform_result}.
Indeed, one easily verifies that, assuming Conditions $(C_1)$--$(C_5)$ are fulfilled, Theorem \ref{thm:nonuniform_result} follows by rescaling the LASSO functional \eqref{eqn:lasso} with the factor $N_T N_R N_t$ and recalling that $\tilde{\sigma} = \sigma / \sqrt{N_T N_R N_t}$.
In order to show that Conditions $(C_1)$--$(C_5)$ hold true with high probability we 
conduct a probabilistic analysis of the extremal singular values of the matrix $\widetilde{\vc{A}}_S$.
Since columns of the matrix $\widetilde{\vc{A}}_S$ belonging to different angle classes (see Section \ref{sec:nonuniform_results} for the definition) are not correlated, it suffices to consider submatrices corresponding to the subsets $S_{[\beta]} \subset S$ (see \eqref{eqn:equivalent_indices}) representing equivalent indices.
In this sense, the following lemma is the main ingredient for the verification of conditions $(C_1)$--$(C_5)$.
The proof of Proposition \ref{prop:opnorm_bound} is rather long and, therefore, postponed until Section \ref{apdx:proof_prop}.

\begin{proposition}\label{prop:opnorm_bound}
 Let $S \subset \GG$ be a support set and let $[\beta]$ be an angle class such that $S_{[\beta]} \subset S$ is not empty.
 Then, for all $\delta \in (0,1)$,
 \[
  \P ( \| \widetilde{\vc{A}}_{S_{[\beta]}}^* \widetilde{\vc{A}}_{S_{[\beta]}} - \ID \|_{2\to 2} \geq \delta ) \leq 7 |S_{[\beta]}| \exp \left( - \frac{\delta \sqrt{N_t}}{8 \sqrt{|S_{[\beta]}|}} \right) .
 \]
\end{proposition}

Before we turn to the verification of Conditions $(C_1)$--$(C_5)$, we derive two corollaries which follow directly from \mbox{Proposition \ref{prop:opnorm_bound}}.
We have already pointed out in \eqref{eqn:matrix_AS_block_diagonal} that, due to the fact that columns of $\widetilde{\vc{A}}$ belonging to different angle classes are not correlated, it holds
\[
 \P ( \opnorm{\widetilde{\vc{A}}_S^* \widetilde{\vc{A}}_S - \ID} \geq \delta ) = \P \big( \max_{[\beta]} \opnorm{\widetilde{\vc{A}}_{S_{[\beta]}}^* \widetilde{\vc{A}}_{S_{[\beta]}} - \ID} \geq \delta \big) .
\]
Using the union bound and the tail bounds provided by Proposition \ref{prop:opnorm_bound} we arrive at
\begin{equation}\label{eqn:union_bound_estimate}
 \P ( \opnorm{\widetilde{\vc{A}}_S^* \widetilde{\vc{A}}_S - \ID} \geq \delta ) \leq 7 \sum_{\betacl} |S_{[ \beta ]}| \exp \left(- \frac{\delta \sqrt{N_t}}{8 \sqrt{ |S_{[ \beta ]}| }} \right) ,
\end{equation}
where the sum is taken over all angle classes $[\beta]$ (with the property that $S_{[\beta]}$ is not empty).
Now it suffices to plug in the inequality $|S_{[\beta]}| \leq \eta |S| / N_R$ which applies for $\eta$-balanced support sets to obtain the following corollary (recall, that there are exactly $N_R$ angle classes $\betacl$).

\begin{corollary}\label{cor:event1}
For each $\eta$-balanced support set $S \subset \GG$ it holds
\[
 \P ( \opnorm{\widetilde{\vc{A}}_S^* \widetilde{\vc{A}}_S - \ID} \geq \delta ) \leq 7 \eta |S| \exp \left( - \frac{\delta \sqrt{N_t N_R}}{8 \sqrt{\eta |S|}} \right).
\]
\end{corollary}

We use the estimate \eqref{eqn:union_bound_estimate} to derive a bound on the coherence of $\widetilde{\vc{A}}$.
Note, that the coherence of this matrix fulfills
\[
 \max_{\varTheta \neq \varTheta'} | \langle \widetilde{\vc{A}}_{\varTheta} , \widetilde{\vc{A}}_{\varTheta'} \rangle | = \max_{\betacl} \max_{\varTheta , \varTheta^\prime \in S_{\betacl} , \varTheta \neq \varTheta'} | \langle \widetilde{\vc{A}}_{\varTheta} , \widetilde{\vc{A}}_{\varTheta'} \rangle | .
\]
Let $\varTheta , \varTheta' \in S_{\betacl}$ be given.
For any matrix $\vc{M}$ and each entry $m_{i,j}$ it holds $|m_{i,j}| \leq \opnorm{\vc{M}}$.
Now, since $\langle \widetilde{\vc{A}}_{\varTheta} , \widetilde{\vc{A}}_{\varTheta'} \rangle$ is an entry of the matrix $\widetilde{\vc{A}}_S^* \widetilde{\vc{A}}_S - \ID$, where we set $S:=\{\varTheta,\varTheta'\}$, we obtain
\[
 \P ( | \langle \widetilde{\vc{A}}_{\varTheta} , \widetilde{\vc{A}}_{\varTheta'} \rangle | \geq u ) \leq 14 \exp \left( - \frac{u \sqrt{N_t}}{8 \sqrt{2}} \right)
\]
from \eqref{eqn:union_bound_estimate}.
Since we can choose $\genfrac(){0pt}{}{N/N_R}{2} \leq (N / N_R)^2 / 2$ pairs $\varTheta , \varTheta'$ and there are $N_R$ angle classes, an application of the union bound yields the following corollary.

\begin{corollary}\label{cor:event2}
The coherence of the matrix $\widetilde{\vc{A}}$ satisfies
\[
 \P ( \max_{\varTheta \neq \varTheta'} | \langle \widetilde{\vc{A}}_{\varTheta} , \widetilde{\vc{A}}_{\varTheta'} \rangle | \geq u )
  \leq \frac{7 N^2}{N_R} \exp \left( - \frac{u \sqrt{N_t}}{8 \sqrt{2}} \right).
\]
\end{corollary}

Next we will derive upper bounds for the respective probabilities that one 
of the conditions $(C_1)$--$(C_5)$ fails to hold.

\subsubsection*{Condition $\vc{(C_1)}$}
We define events $\CC_1^\delta$ as
\[
 \CC_1^\delta = \{ \opnorm{\widetilde{\vc{A}}_{S}^* \widetilde{\vc{A}}_{S} - \ID} \leq \delta \leq 1/2 \} .
\]
Clearly, if $\CC_1^\delta$ occurs (i.e., $\delta \in (0,1/2 ]$), then it holds $\opnorm{(\widetilde{\vc{A}}_{S}^* \widetilde{\vc{A}}_{S})^{-1}} \leq 1 / (1-\delta) \leq 2$.
Hence, Corollary \ref{cor:event1} can be used to show that condition $(C_1)$ is fulfilled with high probability.

\subsubsection*{Condition $\vc{(C_2)}$}
\begin{lemma}\label{lem:a2}
Let $S \subset \GG$ be a $\eta$-balanced support set.
If there exist $\delta \in (0,1)$ and $u>0$ such that
\begin{equation}\label{eqn:event_delta}
 \opnorm{\widetilde{\vc{A}}_{S}^* \widetilde{\vc{A}}_{S} - \ID} \leq \delta , \qquad  \max_{\varTheta \neq \varTheta'} | \langle \widetilde{\vc{A}}_{\varTheta} , \widetilde{\vc{A}}_{\varTheta'} \rangle | \leq u ,
\end{equation}
then it holds
\[
 \P ( \| \widetilde{\vc{A}}^*_{S^c} \widetilde{\vc{A}}_S (\widetilde{\vc{A}}_S^* \widetilde{\vc{A}}_S)^{-1} \sgn (\vc{x}_S)\|_\infty \geq 1/4 ) \leq 2 N \exp \bigg( -\frac{(1-\delta)^2 N_R}{32 u^2 \eta |S|} \bigg).
\]
\end{lemma}
\begin{proof}
For $\varTheta \in S^c$ set $\vc{v}_{\varTheta} := (\widetilde{\vc{A}}_S^* \widetilde{\vc{A}}_S)^{-1} 
\widetilde{\vc{A}}_S^* \widetilde{\vc{A}}_{\varTheta}$.
Since $\sgn (\vc{x}_S)$ forms a Steinhaus sequence, the Hoeffding-type inequality for Steinhaus sequences of 
Lemma~\ref{lem:hoeffding} yields
\begin{equation}\label{eqn:hoeff_bound}
\P ( |\langle \vc{v}_{\varTheta} , \text{sgn}(\vc{x}_S) \rangle| \geq 1 / 4 ) \leq 2 e^{- 1 / ( 32 \|\vc{v}_{\varTheta}\|_2^2)}.
\end{equation}
In order to derive a bound for the norms $\| \vc{v}_{\varTheta} \|_2$ we calculate,
\[
 \| \vc{v}_{\varTheta} \|_2 = \| (\widetilde{\vc{A}}_S^* \widetilde{\vc{A}}_S)^{-1} \widetilde{\vc{A}}_S^* \widetilde{\vc{A}}_{\varTheta} \|_2 \leq \opnorm{(\widetilde{\vc{A}}_S^* \widetilde{\vc{A}}_S)^{-1}}
  \sqrt{ \sum_{\tilde\varTheta \in S} | \langle \widetilde{\vc{A}}_{\tilde\varTheta} , \widetilde{\vc{A}}_{\varTheta} \rangle |^2 }.
\]
By assumption, $\opnorm{(\widetilde{\vc{A}}_S^* \widetilde{\vc{A}}_S)^{-1}}$ is bounded by $(1-\delta)^{-1}$.
Since $S$ is $\eta$-balanced, there are at most $\eta |S| / N_R$ indices $\tilde\varTheta$ such that the summands are not equal to zero (recall that columns corresponding to different angle classes are not correlated).
Together with the fact that each such inner product $| \langle \widetilde{\vc{A}}_{\varTheta'} , \widetilde{\vc{A}}_{\varTheta} \rangle |$ is bounded by $u$, we have
$
 \| \vc{v}_{\varTheta} \|_2 \leq u \sqrt{\eta|S|} / (1-\delta) \sqrt{N_R}
$.
Inserting this estimate into \eqref{eqn:hoeff_bound} and furthermore applying the union bound yields the assertion.
\end{proof}

\subsubsection*{Condition $\vc{(C_5)}$}
\begin{lemma}\label{lem:a5}
If the matrix $\widetilde{\vc{A}}_S$ obeys
$
 \opnorm{\widetilde{\vc{A}}_{S}^* \widetilde{\vc{A}}_{S} - \ID } \leq \delta < 1 ,
$
and $\sgn(\vc{x}_S)$ is a (random) Steinhaus vector,
then it holds 
\[
\P ( \| (\widetilde{\vc{A}}_S^* \widetilde{\vc{A}}_S)^{-1} \sgn (\vc{x}_S) \|_\infty > 3 )
  \leq 2 |S| e^{ - {2(1-\delta)^2} / {\delta^2} }.
\]
\end{lemma}
\begin{proof}
An application of the triangle inequality yields
\[
\| (\widetilde{\vc{A}}_S^* \widetilde{\vc{A}}_S)^{-1} \sgn (\vc{x}_S) \|_\infty \leq
  \| \sgn (\vc{x}_S) \|_\infty + \| ((\widetilde{\vc{A}}_S^* \widetilde{\vc{A}}_S)^{-1} - \ID) \sgn (\vc{x}_S) \|_\infty.
\]
Since the first summand on the right-hand side is equal to $1$, it suffices to show that the second term is bounded by $2$ with high
probability.
Let $\vc{v}_{\varTheta}$, $\varTheta \in S$, denote the rows of the matrix $(\widetilde{\vc{A}}_S^* \widetilde{\vc{A}}_S)^{-1} - \ID$.
Due to the assumption it holds $\opnorm{\ID - \widetilde{\vc{A}}_S^* \widetilde{\vc{A}}_S} \leq \delta$.
Thus,
\[
 \| \vc{v}_{\varTheta} \|_2 = \| ( (\widetilde{\vc{A}}_S^* \widetilde{\vc{A}}_S)^{-1} - \ID ) \vc{e}_{\varTheta} \|_2 = \bigg\| \sum_{k=1}^\infty ( \ID - \widetilde{\vc{A}}_S^* \widetilde{\vc{A}}_S )^k \vc{e}_{\varTheta} \bigg\|_2 \leq \frac{1}{1-\delta} \opnorm{\ID - \widetilde{\vc{A}}_S^* \widetilde{\vc{A}}_S} \leq \frac{\delta}{1-\delta} .
\]
Using Hoeffding's inequality for Steinhaus sums (Lemma \ref{lem:hoeffding}) we obtain
\[
\P ( | \langle \vc{v}_{\varTheta} , \sgn (\vc{x}_S) \rangle | \geq 2 ) \leq 2 e^{ - 2 / \|\vc{v}_{\varTheta}\|^2 }
  \leq 2 e^{ - 2 (1-\delta)^2 / \delta^2} .
\]
Finally, the assertion follows by applying the union bound.
\end{proof}

\subsubsection*{Condition $\vc{(C_3)}$}
\begin{lemma}\label{lem:a3}
If the matrix $\widetilde{\vc{A}}_S$ obeys $\opnorm{\widetilde{\vc{A}}_{S}^* \widetilde{\vc{A}}_{S} - \ID} \leq 1/2$, then,
for a mean-zero Gaussian vector $\widetilde{\vc{n}}$ with independent entries of variance  $\tilde{\sigma}^2$ and 
$\mu = \tilde\sigma \sqrt{2\log (N)}$, it holds,
\[
 \P_{{\vc{n}}} ( \|(\widetilde{\vc{A}}_S^* \widetilde{\vc{A}}_S)^{-1} \widetilde{\vc{A}}_S^* \widetilde{\vc{n}}\|_\infty \geq 2 \mu ) \leq N^{-3} .
\]
\end{lemma}
\begin{proof}
For $\varTheta \in S$ we write $\vc{v}_{\varTheta}$ to denote the row of the matrix $(\widetilde{\vc{A}}_S^* \widetilde{\vc{A}}_S)^{-1} \widetilde{\vc{A}}_S^*$ corresponding to the index $\varTheta$.
An application of Lemma \ref{lem:inner_normal_bound} gives, since $\tilde\sigma^{-1} \widetilde{\vc{n}}$ is a standard complex Gaussian random vector,
\[
 \P_{{\vc{n}}} ( |\langle \vc{v}_{\varTheta} , \widetilde{\vc{n}} \rangle| \geq 2 \mu ) = \P_{{\vc{n}}} ( |\langle \vc{v}_{\varTheta} , {\tilde{\sigma}}^{-1} \widetilde{\vc{n}} \rangle| \geq 2 \sqrt{2\log (N)} ) \leq e^{- 8 \log (N) / \|\vc{v}_{\varTheta}\|_2^2 } .
\]
Since $\vc{v}_{\varTheta}$, $\varTheta \in S$, is a column of the matrix
$\vc{M} := \widetilde{\vc{A}}_S (\widetilde{\vc{A}}_S^* \widetilde{\vc{A}}_S)^{-1}$, the operator norm of $\vc{M}$ provides a
uniform bound for the $\ell_2$-norms of the rows $\vc{v}_{\varTheta}$.
By assumption $\opnorm{\widetilde{\vc{A}}_{S}^* \widetilde{\vc{A}}_{S} - \ID} \leq 1/2$, and, hence, 
\[
\|\vc{v}_{\varTheta}\|_2^2 \leq \opnorm{\vc{M}}^2 = \opnorm{\vc{M}^* \vc{M}} = \opnorm{(\widetilde{\vc{A}}_S^* \widetilde{\vc{A}}_S)^{-1}} \leq (1-1/2)^{-1} = 2.
\]
Combining this with the tail bound above yields
\[
 \P_{{\vc{n}}} ( |\langle \vc{v}_{\varTheta} , \widetilde{\vc{n}} \rangle| \geq 2 \mu ) \leq e^{- 8 \log (N) / 2 } = N^{-4} .
\]
An application of the union bound, using $|S| \leq N$, implies the assertion.
\end{proof}

\subsubsection*{Condition $\vc{(C_4)}$}
In order to deal with Condition $(C_4)$, we need an estimate in the $\ell_2$-norms of the columns of the \mbox{matrix $\widetilde{\vc{A}}$}.

\begin{lemma}\label{lem:l2_norm_columns}
The $\ell_2$-norms of the columns of the matrix $\widetilde{\vc{A}}$ satisfy
\begin{equation*}
 \P ( \max_{\varTheta \in \GG} \| \widetilde{\vc{A}}_\varTheta \|_2 \geq 2/ \sqrt{3} ) \leq N e^{-(2/\sqrt{3}-1)^2 N_t} .
\end{equation*}
\end{lemma}
\begin{proof}
It follows from the definition of the columns $\widetilde{\vc{A}}$ in \eqref{eqn:columns} that $\| \widetilde{\vc{A}}_\varTheta \|_2$,
$\Theta \in \GG$, is given by $\frac{1}{\sqrt{N_t}} \|\vc{z}\|_2$ with $\vc{z} = \frac{1}{\sqrt{N_T}} \sum_{i=1}^{N_T} \ee{d_T \beta \Delta_\beta (i-1)} \vc{s}_i$, where the $\vc{s}_i$ are independent standard complex Gaussian vectors.
Thus, also $\vc{z}$ is an $N_t$-dimensional standard complex Gaussian random vector so that 
Lemma~\ref{lem:gaus_vec_tail} yields
\[
 \P ( \| \widetilde{\vc{A}}_\varTheta \|_2 \geq 2/ \sqrt{3} ) = \P (  \| \vc{z} \|_2 \geq \sqrt{N_t} + (2/ \sqrt{3}-1) \sqrt{N_t} ) \leq e^{-(2/ \sqrt{3}-1)^2 N_t} .
\]
The assertion follows by applying the union bound.
\end{proof}

Under the assumption that the columns are bounded in the $\ell_2$-norm, we are able to provide an estimate 
for the probability that condition $(C_4)$ fails to hold.

\begin{lemma}\label{lem:a4}
Let $\widetilde{\vc{n}}$ be a mean-zero Gaussian vector with covariance matrix $\tilde\sigma^2 \vc{\operatorname{Id}}$ and
$\mu = \tilde{\sigma} \sqrt{2 \log(N)}$. If $\| \widetilde{\vc{A}}_\varTheta \|_2 \leq 2/ \sqrt{3}$ for all $\varTheta \in S^c$, then
\[
 \P_{{\vc{n}}} ( \| \widetilde{\vc{A}}^*_{S^c} (\ID - \Pi_S) \widetilde{\vc{n}} \|_\infty \geq \sqrt{2} \mu ) \leq N^{-2} .
\]
\end{lemma}
\begin{proof}
Note that the $\ell_2$-norms of the rows of the matrix $\widetilde{\vc{A}}^*_{S^c} (\ID - \Pi_S)$ are uniformly bounded by
\[
 \max_{\varTheta \in S^c} \| (\ID - \Pi_S) \widetilde{\vc{A}}_\varTheta \|_2 \leq \underbrace{\opnorm{\ID - \Pi_S}}_{\leq 1} \max_{\varTheta \in S^c} \| \widetilde{\vc{A}}_\varTheta \|_2 \leq 2/ \sqrt{3} .
\]
Using the same arguments as in the proof of Lemma \ref{lem:a3} (now with $\vc{v}_{\varTheta}$, $\varTheta \in S^c$, denoting the columns of the matrix $(\ID - \Pi_S) \widetilde{\vc{A}}_\varTheta$), we obtain
\[
 \P_{{\vc{n}}} ( |\langle \vc{v}_{\varTheta} , \widetilde{\vc{n}} \rangle| \geq \sqrt{2} \mu ) \leq e^{-4 \log(N) / \| \vc{v}_{\varTheta} \|_2^2} \leq e^{-3 \log(N)} = N^{-3} ,
\]
where, in the last step, we used the bound on the norm $\| \vc{v}_{\varTheta} \|_2$ from above.
An application of the union bound, using that $|S^c| \leq N$, implies the assertion.
\end{proof}

\subsubsection{Proof of Theorem \ref{thm:nonuniform_result}\texorpdfstring{\color{red}(a)}{(a)}}\label{sec:proof_nonuniform_result_a}
Now we are prepared to finalize the proof of Theorem \ref{thm:nonuniform_result}.
In view of Lemma~\ref{lem:candes}, we only have to verify that Conditions $(C_1)$--$(C_5)$ are fulfilled with probability exceeding $1 - 7 \max \{\varepsilon , N^{-2} \}$ in order to show \mbox{part (a)} of Theorem \ref{thm:nonuniform_result}.
As we have already noted, Condition $(C_1)$ holds, if $\opnorm{\widetilde{\vc{A}}_{S}^* \widetilde{\vc{A}}_{S} - \ID} \leq \delta$, for some $\delta \in (0,1/2 ]$.
Thus, Corollary \ref{cor:event1} can be used to show that $(C_1)$ is fulfilled with high probability.
Moreover, a careful inspection of Corollaries \ref{cor:event1}, \ref{cor:event2} and Lemmas \ref{lem:a2}--\ref{lem:a4} reveals that the probability of the assertion of Theorem \ref{thm:nonuniform_result} \textit{failing} to hold true is bounded by the probability
$
 \P ( \neg \AA_{\delta , u} ) + \sum_{k=2}^5 \P ( \neg \CC_k | \AA_{\delta , u} ),
$
where each $\CC_k$ stands for the event when Condition $(C_k)$ is fulfilled, and the event $\AA_{\delta , u}$ is defined as
\begin{equation}\label{eqn:AA_delta_u}
 \AA_{\delta , u} = \underbrace{\big\{ \opnorm{\widetilde{\vc{A}}_S^* \widetilde{\vc{A}}_S - \ID} < \delta \big\}}_{=\CC_1^\delta} \cap \underbrace{\big\{ \max_{\varTheta \neq \varTheta'} | \langle \widetilde{\vc{A}}_{\varTheta} , \widetilde{\vc{A}}_{\varTheta'} \rangle | \leq u \big\}}_{=:\EE_1^u} \cap \underbrace{\big\{ \max_{\varTheta \in \GG} \| \widetilde{\vc{A}}_\varTheta \|_2 \leq 2/ \sqrt{3} \big\}}_{=:\EE_2} .
\end{equation}
Now, Corollaries \ref{cor:event1}, \ref{cor:event2} and Lemma \ref{lem:l2_norm_columns} imply
\begin{align}
 \P ( \neg \AA_{\delta , u} )
 &   \leq \P ( \neg \CC_1^\delta ) + \P ( \neg \EE_1^u ) + \P ( \neg \EE_2 ) \notag   \\
 &   \leq 7 \eta |S| \exp \left( -\frac{\delta \sqrt{N_R N_t}}{8 \sqrt{\eta |S|}} \right) + \frac{7N^2}{N_R} \exp \left( -\frac{u\sqrt{N_t}}{8\sqrt{2}} \right) + N e^{(2/ \sqrt{3} - 1)^2 N_t} . \label{eqn:prob_summands1}
\end{align}
Due Lemmas \ref{lem:a2}, \ref{lem:a5}, \ref{lem:a3}, and \ref{lem:a4} we have the following bounds on the failure probability of Conditions $(C_2)$, $(C_5)$, $(C_3)$, and $(C_4)$ --- conditional on the event $\AA_{\delta , u}$:
\begin{align}
 \P ( \neg \CC_2 | \AA_{\delta , u} ) &\leq 2N e^{-(1-\delta)^2 N_R / 32u^2 \eta |S|} ,
  & \P ( \neg \CC_5 | \AA_{\delta , u} ) &\leq  2 |S| e^{-2(1-\delta)^2 / \delta^2} , \label{eqn:prob_summands2}  \\
 \P ( \neg \CC_3 | \AA_{\delta , u} ) &\leq N^{-3} ,
  & \P ( \neg \CC_4 | \AA_{\delta , u} ) &\leq N^{-2} . \label{eqn:prob_summands3}
\end{align}
Since we claim that the assertion of Theorem \ref{thm:nonuniform_result} holds true with probability exceeding $1 - 7 \max \{ \varepsilon , N^{-2} \}$, it is enough to verify that each of the seven terms in \eqref{eqn:prob_summands1}, \eqref{eqn:prob_summands2}, \eqref{eqn:prob_summands3} is either bounded by $\varepsilon$ or by $N^{-2}$.
Since this holds trivially for the terms in \eqref{eqn:prob_summands3}, we are left with the terms in \eqref{eqn:prob_summands1} and \eqref{eqn:prob_summands2} which, as we will see, can all be bounded by $\varepsilon$.
Due to the definition of the parameter $\eta$ (cf. Definition \ref{def:balanced}) it holds (note that each equivalence class contains $N/N_R$ elements)
\begin{equation*}
 \eta = \max_{[\beta]} \frac{|S_{[\beta]}|}{|S|/N_R} \leq \frac{N/N_R}{|S|/N_R} = \frac{N}{|S|} .
\end{equation*}
Thus, $\eta |S| \leq N$, which yields that in order to verify that the first term in \eqref{eqn:prob_summands1} is smaller or equal to $\varepsilon$ it is sufficient to show that
\begin{equation}\label{eqn:scaling_1}
\frac{64 \eta |S|}{\delta^2} \log^2 \bigg( \frac{7N}{\varepsilon} \bigg) \leq N_R N_t,
\end{equation}
whereas for the second term in \eqref{eqn:prob_summands1} to be smaller or equal to $\varepsilon$ it suffices that
\begin{equation}\label{eqn:scaling_2}
\frac{128 N_R}{u^2} \log^2 \bigg( \frac{7 N^2}{\varepsilon N_R} \bigg) \leq N_R N_t .
\end{equation}
The term $\log ( 7 N^2 / \varepsilon N_R )$ is dominated by $2 \log ( 7 N / \varepsilon )$ and, therefore, we can can set 
\begin{equation}\label{eqn:u}
 u := \delta \sqrt{8 N_R} / \sqrt{\eta |S|}
\end{equation}
in order to ensure that \eqref{eqn:scaling_2} is already implied by \eqref{eqn:scaling_1}. 
In order to show \eqref{eqn:scaling_1}, we set
\begin{equation*}
 \delta = \frac{1}{\sqrt{K \log (eN/\varepsilon)}} ,
\end{equation*}
where $K>0$ will be chosen below.
With this the left-hand side of \eqref{eqn:scaling_1} can be estimated as
\[
 \frac{64 \eta |S|}{\delta^2} \log^2 \bigg( \frac{7N}{\varepsilon} \bigg) = 64K \eta |S| \log \bigg( \frac{e N}{\varepsilon} \bigg) \log^2 \bigg( \frac{7N}{\varepsilon} \bigg) \leq 64K \log^3 (7) \eta |S| \log^3 \bigg( \frac{N}{\varepsilon} \bigg) .
\]
Due to the assumption \eqref{eqn:nonuniform_scaling} we have
\[
 N_R N_t \geq C \eta |S| \log^3 \bigg( \frac{N}{\varepsilon} \bigg) ,
\]
where $C > 0$ is a suitable constant. 
By setting $C = 64K\log^3 (7)$, we establish \eqref{eqn:scaling_1} (and, hence, also \eqref{eqn:scaling_2}) which means that both the first and the second term in \eqref{eqn:prob_summands1} are bounded by $\varepsilon$.
In order to ensure that the last term $N e^{-(2/\sqrt{3}-1)^2 N_t}$ in \eqref{eqn:prob_summands1} is bounded by the second one (and, hence, is also smaller or equal to $\varepsilon$), it suffices to have $(2/\sqrt{3}-1)^2 N_t \geq u \sqrt{N_t}/8 \sqrt{2}$ which, recalling the definition of $u$, is equivalent to
\[
 \delta \leq 4 (2/\sqrt{3}-1)^2 \sqrt{\frac{\eta |S| N_t}{N_R}} .
\]
Since by definition of the parameter $\eta$ we have $\eta |S| / N_R \geq |S_{[\beta]}| \geq 1$ and, moreover, $N_t \geq 1$, anyway, this latter condition is fulfilled as long as $\delta \leq 4 (2/\sqrt{3}-1)^2$ which is fulfilled whenever $K$ is sufficiently large.
Since $\log ( eN / \varepsilon) \geq 1$, and, hence, $\delta \leq 1 / \sqrt{K}$, we have for the first term in \eqref{eqn:prob_summands2}, recalling that $u^2 = 8 \delta^2 N_R / \eta |S|$,
\begin{align*}
 & 2 N \exp \left( -\frac{[1-\delta]^2 N_R}{32 u^2 \eta |S|} \right)
  = 2 N \exp \left( -\frac{[1-\delta]^2}{256 \delta^2} \right)  \\
  & = 2 N \exp \left( -\frac{1}{256} \left[1-\frac{1}{\sqrt{K \log (eN/\varepsilon)}}\right]^2 K \log \left( \frac{eN}{\varepsilon} \right) \right)  \\
  &\leq 2 N \exp \left( -\frac{1}{256} \left[1 - \frac{1}{\sqrt{K}} \right]^2 K \log \left( \frac{2N}{\varepsilon} \right) \right)  
  = 2 N \exp \left( -\frac{1}{256} \left[\sqrt{K}-1 \right]^2 \log \left( \frac{2N}{\varepsilon} \right) \right) .
\end{align*}
Therefore, in order to show that the first term in \eqref{eqn:prob_summands2} is bounded by $\varepsilon$, it is sufficient to ensure that $(\sqrt{K}-1)^2 \geq 256 = 16^2$ which holds true if $\sqrt{K} \geq 17$.
The second term in \eqref{eqn:prob_summands2} is bounded by the first one.
Finally, the condition $\delta \leq 1/2$ is satisfied since $\delta \leq 1 / \sqrt{K} \leq 1/17$.
This finishes the proof Theorem \ref{thm:nonuniform_result}{\color{red}(a)}.\qed

\subsubsection{Proof of Theorem \ref{thm:nonuniform_result}\texorpdfstring{\color{red}(b)}{(b)}}\label{sec:proof_nonuniform_result_b}
For part (b) it remains to derive the stated approximation property of the solution to the debiased LASSO with respect to $\vc{x}_S$.
From the proof of part (a) it follows that, with high probability (as stated in the theorem), the columns of the matrix $\vc{A}_S$ are linearly independent.
Hence, the solution $\vc{z}^\#$ to the least squares problem \eqref{eqn:debiasing} is given by $\vc{z}^\# = (\vc{A}_S^* \vc{A}_S)^{-1} \vc{A}_S^* \vc{y}$.
Indeed, this formula can easily be verified by calculating the corresponding optimality conditions.
Since $\vc{x}$ is exactly sparse and supported on the set $S$ it holds $\vc{y} = \vc{A}_S \vc{x}_S + \vc{n}$, i.e.,
\[
 \vc{z}^\# = \vc{x}_S + (\vc{A}_S^* \vc{A}_S)^{-1} \vc{A}_S^* \vc{n} .
\]
Now, by a standard inequality between the $\ell_2$ and the $\ell_\infty$-norm,
\begin{equation}\label{eqn:approx_debiase}
 \| \vc{z}^\# - \vc{x}_S \|_2 = \| (\vc{A}_S^* \vc{A}_S)^{-1} \vc{A}_S^* \vc{n} \|_2 \leq \sqrt{s} \| (\vc{A}_S^* \vc{A}_S)^{-1} \vc{A}_S^* \vc{n} \|_\infty .
\end{equation}
Since we may assume that we are in the event where condition ($\text{C}_3$) holds true, we have the estimate
\[
 \| (\vc{A}_S^* \vc{A}_S)^{-1} \vc{A}_S^* \vc{n} \|_\infty = \| (\widetilde{\vc{A}}_S^* \widetilde{\vc{A}}_S)^{-1} \widetilde{\vc{A}}_S^* \widetilde{\vc{n}} \|_\infty \leq 2 \mu = 2 \tilde{\sigma} \sqrt{2\log (N)} = \frac{2 \sigma}{\sqrt{N_T N_R N_t}} \sqrt{2\log (N)} ,
\]
where, in the second but last step, we plugged in the definition of $\mu$ from \eqref{eqn:scaled_lasso}.
A combination of this latter inequality with \eqref{eqn:approx_debiase} implies the desired approximation property.\qed

\subsection{Recovery via basis pursuit denoising under sparsity defect}\label{sec:sparsity_defect_result}
We replace the original constraint $\| \vc{A}\vc{x} - \vc{y} \|_2 \leq \varrho$ of the basis pursuit denoising program \eqref{eqn:basis_pursuit} by the equivalent constraint $\|\widetilde{\vc{A}} \vc{x} - \widetilde{\vc{y}}\|_2 \leq \tilde \varrho$, where $\tilde\varrho = \varrho / \sqrt{N_T N_R N_t}$.
For the proof of Theorem \ref{thm:nonuniform_result_sparsity_defect} we use a general recovery result for basis pursuit denoising stated as \mbox{Theorem \ref{thm:bp_recovery_via_inexact_dual}} below.
The original result \cite[Thm. 4.33]{Foucart2013} uses an inexact dual certificate which in our case can actually be taken as the canonical exact dual certificate given by
\[
 \widetilde{\vc{h}} = \widetilde{\vc{A}}_S (\widetilde{\vc{A}}_S^* \widetilde{\vc{A}}_S)^{-1} \sgn (\vc{x}_S) .
\]
Hence, the following Theorem \ref{thm:bp_recovery_via_inexact_dual}, as it is formulated here, is a special case of \cite[Thm. 4.33]{Foucart2013}, see also \cite{hürast14}.
The constants $C_1$ and $C_2$ appearing in \eqref{eqn:approx_prop} below can be calculated explicitly using the corresponding formulas in \mbox{\cite[Thm. 4.33]{Foucart2013}}.
However, we refrain from doing so at this point.

\begin{theorem}[\mbox{cf. \cite[Thm. 4.33]{Foucart2013}}]\label{thm:bp_recovery_via_inexact_dual}
 Let $\vc{x} \in \C^N$ be given with $s$ largest absolute entries supported on $S \subset [N]$.
 Assume that measurements $\widetilde{\vc{y}}=\widetilde{\vc{A}}\vc{x} + \widetilde{\vc{n}} \in \C^m$ are given with $m<N$ and $\|\widetilde{\vc{n}}\|_2 \leq \tilde\varrho$.
 If
 \[
  \opnorm{\widetilde{\vc{A}}_S^* \widetilde{\vc{A}}_S - \ID} \leq \delta < 1 , \qquad \| \widetilde{\vc{A}}_{S^c}^* \widetilde{\vc{A}}_S (\widetilde{\vc{A}}_S^* \widetilde{\vc{A}}_S)^{-1} \sgn (\vc{x}_S) \|_\infty < 1 ,
 \]
 and, moreover, for $\alpha_1 , \alpha_2 \geq 0$,
 \begin{equation}\label{eqn:inexact_dual_assumption}
  \max_{\ell \in S^c} \|\widetilde{\vc{A}}_S^* \widetilde{\vc{A}}_\ell \|_2 \leq \alpha_1 , \qquad \| \widetilde{\vc{A}}_S (\widetilde{\vc{A}}_S^* \widetilde{\vc{A}}_S)^{-1} \sgn (\vc{x}_S) \|_2 \leq \alpha_2 \sqrt{s} ,
 \end{equation}
 then each solution $\vc{x}^\#$ to the basis pursuit denoising program
 \[
  \min_{\vc{z}} \| \vc{z} \|_1 \quad \text{subject to} \quad \| \widetilde{\vc{A}}\vc{z} - \widetilde{\vc{y}} \|_2 \leq \tilde\varrho
 \]
 satisfies
 \begin{equation}\label{eqn:approx_prop}
  \|\vc{x}^\# - \vc{x} \|_2 \leq C_1 \inf_{\text{$s$-sparse $\vc{x}'$}} \| \vc{x}' - \vc{x} \|_1 + C_2 \sqrt{s} \tilde\varrho ,
 \end{equation}
 for some constants $C_1$ and $C_2$ depending only on $\delta , \alpha_1 , \alpha_2$.
\end{theorem}

\begin{remark} The error estimate \eqref{eqn:approx_prop} is slightly weaker as the one that holds under the RIP, see \eqref{eqn:approx}. In fact,  \eqref{eqn:approx_prop} features the $\ell_2$-norm on the right hand side, but
\eqref{eqn:approx} states an error estimate in the $\ell_1$-norm. This seems the price one has to pay when not working
with strong recovery conditions such as the RIP.
\end{remark}

\goodbreak
\subsubsection{Proof of Theorem \ref{thm:nonuniform_result_sparsity_defect}}\label{sec:proof_nonuniform_result_sparsity_defect}
In order to prove Theorem \ref{thm:nonuniform_result_sparsity_defect} we show that the assumptions of 
Theorem~\ref{thm:bp_recovery_via_inexact_dual} are fulfilled with high probability.
This can be achieved by simply repeating some of the arguments we already used during the proof of 
Theorem~\ref{thm:nonuniform_result}.
Indeed, recalling the event $\AA_{\delta , u}$ (see \eqref{eqn:AA_delta_u}) from the proof of Theorem \ref{thm:nonuniform_result}, where we now choose $\delta = 1/2$ and $u$ as in \eqref{eqn:u}, then it holds (for details see the proof of Theorem \ref{thm:nonuniform_result_sparsity_defect})
\[
 \AA_{\delta , u} \subset \{ \opnorm{\widetilde{\vc{A}}_S^* \widetilde{\vc{A}}_S - \ID} \leq 1/2 \} \cap \{ \max_{\varTheta \in \GG} \| \widetilde{\vc{A}}_\varTheta \|_2 \leq 2/ \sqrt{3} \} .
\]
This means that, in case of the event $\AA_{\delta ,u}$ it holds
\[
 \max_{\varTheta \in S^c} \|\vc{A}_S^* \vc{A}_\varTheta \|_2 \leq \opnorm{\vc{A}_S^*} \max_{\varTheta \in S^c} \| \vc{A}_\varTheta \|_2 \leq 2 \sqrt{1+1/2} / \sqrt{3} = \sqrt{2} ,
\]
and, moreover,
\[
 \| \vc{A}_S (\vc{A}_S^* \vc{A}_S)^{-1} \sgn (\vc{x}_S) \|_2 \leq \opnorm{\vc{A}_S} \opnorm{(\vc{A}_S^* \vc{A}_S)^{-1}} \| \sgn (\vc{x}_S) \|_2 \leq \frac{\sqrt{1+\delta}}{1-\delta} \sqrt{s} = \sqrt{6} \sqrt{s} .
\]
As a consequence, by setting $\alpha_1 = \sqrt{2}$ and $\alpha_2 = \sqrt{6}$ the assumptions in \eqref{eqn:inexact_dual_assumption} are automatically fulfilled, provided that the event $\AA_{\delta ,u}$ occurs.
Finally, as the proof of Theorem~\ref{thm:nonuniform_result} shows, it holds $\P ( \neg \AA_{\delta , u} ) \leq 3\varepsilon$ and
(by Lemma~\ref{lem:a2})
\[
 \P ( \| \widetilde{\vc{A}}_{S^c}^* \widetilde{\vc{A}}_S (\widetilde{\vc{A}}_S^* \widetilde{\vc{A}}_S)^{-1} \sgn (\vc{x}_S) \|_\infty \geq 1/4 \mid \AA_{\delta , u} ) \leq \varepsilon .
\]
This is also a consequence of the special choice of $u$, depending on $\delta$, and the condition \eqref{eqn:nonuniform_scaling2} on the number of targets $s$ from the assumptions of Theorem \ref{thm:nonuniform_result_sparsity_defect}.
Thus, using the union bound, the probability that the assumptions of Theorem~\ref{thm:bp_recovery_via_inexact_dual} \textit{fail} to hold are bounded by
\[
 \P ( \neg \AA_{\delta ,u} ) + \P ( \| \widetilde{\vc{A}}_{S^c}^* \widetilde{\vc{A}}_S (\widetilde{\vc{A}}_S^* \widetilde{\vc{A}}_S)^{-1} \sgn (\vc{x}_S) \|_\infty \geq 1/4 \mid \AA_{\delta , u} ) \leq 3\varepsilon+ \varepsilon = 4\varepsilon. \hfill\qed
\]

\subsection{Proof of Proposition \ref{prop:opnorm_bound}}\label{apdx:proof_prop}
We now show Proposition~\ref{prop:opnorm_bound} concerning the well-conditionedness of the submatrix 
$\widetilde{\vc{A}}_{S_{[\beta]}}$.
In what follows we write $s_{(i,a)}$ for the $a$-th entry of the signal vector $\vc{s}_i$.
In order to analyze the quantity $\|\widetilde{\vc{A}}_{S_{[\beta]}}^* \widetilde{\vc{A}}_{S_{[\beta]}} - \ID\|_{2\to 2}$ 
we write the matrix product $\widetilde{\vc{A}}_{S_{[\beta]}}^* \widetilde{\vc{A}}_{S_{[\beta]}}$ as 
(see Appendix \ref{apdx:section3}, \eqref{eqn:AA_repr_2} for a proof)
\begin{equation}\label{eqn:AA_repr}
 \widetilde{\vc{A}}_{S_{[\beta]}}^* \widetilde{\vc{A}}_{S_{[\beta]}} = \sum_{i,j = 1}^{N_T} \sum_{a,b = 1}^{N_t} \overline{{s}_{(i,a)}} {s}_{(j,b)} \vc{Y}^{(i,a),(j,b)},
\end{equation}
where $\vc{Y}^{(i,a),(j,b)}$ is an $|S_{[\beta]}| \times |S_{[\beta]}|$ matrix which, for indices $\varTheta , \varTheta' \in S_{[\beta]}$, $\varTheta = (\beta,\tau,f)$, $\varTheta^\prime = (\beta^\prime,\tau^\prime,f^\prime)$, is given entrywise by
\begin{equation}\label{eqn:def_B_matrix_in_sum}
 [\vc{Y}^{(i,a),(j,b)}]_{\varTheta,\varTheta'} = \delta^{\sim_{N_t}}_{a-b , \tau^\prime -\tau} (N_T N_t)^{-1} \ee{\frac{f'-f}{N_t} (\tau +a-1)} \ee{d_T \Delta_\beta [\beta^\prime (j-1)-\beta (i-1)]} ,
\end{equation}
and where $\delta^{\sim_{N_t}}_{\cdot , \cdot}$ is defined in \eqref{eqn:general_kronecker}.
By subdividing the sum in \eqref{eqn:AA_repr} into an off-diagonal  part containing the summands with $(i,a)\neq (j,b)$
and a remaining diagonal part, and then applying the triangle inequality, it follows that
\begin{align}
 & \big\| \widetilde{\vc{A}}_{S_{[\beta]}}^* \widetilde{\vc{A}}_{S_{[\beta]}} - \ID \big\|_{2\to 2} \notag\\
  &\leq  \bigg\| \bigg[ \sum_{i = 1}^{N_T} \sum_{a = 1}^{N_t} |{s}_{(i,a)}|^2 \vc{Y}^{(i,a),(i,a)} \bigg] - \ID \bigg\|_{2\to 2} + \bigg\| \mathop{\sum_{i,j=1}^{N_T} \sum_{a,b=1}^{N_t}}_{(i,a) \neq (j,b)} \overline{{s}_{(i,a)}} {s}_{(j,b)} \vc{Y}^{(i,a),(j,b)} \bigg\|_{2\to 2}  \notag \\
  &=: \opnorm{\vc{\YY}_{=} - \ID} + \opnorm{\vc{\YY}_{\neq}} .  \label{eqn:tail_triangle}
\end{align}

\subsubsection*{A tail bound for $\opnorm{\vc{\YY}_{=} - \ID}$}
By plugging in the definition of the entries $[\vc{Y}^{(i,a),(i,a)}]_{\varTheta,\varTheta^\prime}$ from \eqref{eqn:def_B_matrix_in_sum} one finds
\begin{equation*}
 \sum_{i=1}^{N_T} \sum_{a=1}^{N_t} [\vc{Y}^{(i,a),(i,a)}]_{\varTheta,\varTheta^\prime} = \delta_{\tau^\prime ,\tau} (N_T N_t)^{-1} \ee{\frac{f'-f}{N_t} \tau} \sum_{a=1}^{N_t} \ee{\frac{f'-f}{N_t} (a-1)} \sum_{i=1}^{N_T} \ee{d_T \Delta_\beta (\beta^\prime -\beta ) (i-1)} = \delta_{\varTheta^\prime ,\varTheta} ,
\end{equation*}
i.e., the sum over all matrices $\vc{Y}^{(i,a),(i,a)}$ is the identity matrix.
Therefore, we can write the matrix $\vc{\YY}_{=} - \ID$ in \eqref{eqn:tail_triangle} as a sum over random matrices, namely
\begin{equation}\label{eqn_YYgleich}
 \vc{\YY}_{=} - \ID 
 = \sum_{i = 1}^{N_T} \sum_{a = 1}^{N_t} \widehat{\vc{Y}}^{(i,a)} , \qquad \widehat{\vc{Y}}^{(i,a)} = \chi_{(i,a)} \vc{Y}^{(i,a),(i,a)} , \qquad \chi_{(i,a)} = |{s}_{(i,a)}|^2 - 1 .
\end{equation}
The random matrices $\widehat{\vc{Y}}^{(i,a)}$ have 
block diagonal structure.
To see this, we partition the index set $S_{[\beta]}$ as 
\begin{equation*}
 S_{[\beta]} = \bigcup_{\tau \in [N_t]} S_{[\beta]}^\tau , \qquad S_{[\beta]}^\tau = \{ (\beta^\prime , \tau^\prime , f^\prime ) \in S_{[\beta]} \midcol \tau^\prime = \tau \} .
\end{equation*}
Recall that the matrix $\vc{Y}^{(i,a),(i,a)}$ is given entrywise by
\[
 [\vc{Y}^{(i,a),(i,a)}]_{\varTheta,\varTheta'} = \delta_{\tau , \tau^\prime} (N_T N_t)^{-1} \ee{\frac{f'-f}{N_t} (\tau + a -1)} \ee{d_T \Delta_\beta (\beta^\prime -\beta) (i-1)} .
\]
The appearance of the factor $\delta_{\tau , \tau^\prime}$ implies that the matrices $\vc{Y}^{(i,a),(i,a)}$ (and, hence, also 
$\widehat{\vc{Y}}^{(i,a)}$) have a block diagonal structure. In case of the matrices $\widehat{\vc{Y}}^{(i,a)}$ one obtains --- 
enumerating the indices $(\beta,\tau,f)\in S_{[\beta]}$ adequately ---
\begin{equation}\label{eqn:block_struct}
 \widehat{\vc{Y}}^{(i,a)} = \chi_{(i,a)} \left[ \begin{array}{llll}
                                \vc{\Lambda}_1^{(i,a)} & &  \\
                                 & \vc{\Lambda}_2^{(i,a)} & &  \\
                                 & & \ddots &  \\
                                 & & & \vc{\Lambda}_{N_t}^{(i,a)}
                               \end{array} \right] , \qquad \vc{\Lambda}_{\tau}^{(i,a)} = \vc{v}_{\tau}^{(i,a)} \big[ \vc{v}_{\tau}^{(i,a)} \big]^* ,
\end{equation}
where the vectors $\vc{v}_{\tau}^{(i,a)} \in \C^{|S_{[\beta]}^\tau |}$ are given entrywise as
\[
 \big[ \vc{v}_{\tau}^{(i,a)} \big]_{(\beta,\tau,f)} = (N_T N_t)^{-{1}/{2}} \ee{\frac{f}{N_t} (\tau +a-1)} \ee{d_T \Delta_\beta \beta (i-1)} , \qquad (\beta,\tau,f) \in S_{[\beta]}^\tau .
\]
Note, since for particular choices of $\tau \in [N_t]$ the index sets $S_{[\beta]}^\tau$ might be empty, it is possible that some blocks $\vc{\Lambda}_{\tau}^{(i,a)}$ appearing in \eqref{eqn:block_struct} might be zero-dimensional.
In order to analyze the random variables in \eqref{eqn_YYgleich}, we use the following result, taken from \cite{tr12}.
Below we write $\lambda_{\text{max}} (\cdot )$ to denote the maximal eigenvalue of a given (self-adjoint) matrix.

\begin{theorem}[cf. \mbox{\cite[Thm. 6.2]{tr12}}]\label{thm:tropp}
 Consider a finite sequence $\{ \vc{X}_k \}$ of independent, random, self-adjoint $d \times d$ matrices.
 Assume that
 \begin{equation}\label{eqn:tropp_moments}
  \E \vc{X}_k = 0 \quad \text{and} \quad \E \vc{X}_k^p \preccurlyeq \frac{p!}{2} R^{p-2} \vc{A}_k^2 , \quad \text{for } p = 2 , 3 , 4 , \ldots
 \end{equation}
 Compute the variance parameter
 \begin{equation}\label{eqn:tropp_variance}
  \sigma^2 = \big\| \sum_{k} \vc{A}_k^2 \big\|_{2\to 2} .
 \end{equation}
 Then the following chain of inequalities holds for all $t \geq 0$:
 \begin{equation}\label{eqn:tropp_tail_bounds}
  \P \bigg( \lambda_{\text{max}} \bigg( \sum_{k} \vc{X}_k \bigg) \geq t \bigg) \leq d \exp \bigg( - \frac{t^2 /2}{\sigma^2 + Rt} \bigg) \leq
  \begin{cases}
   d \exp (-t^2 / 4\sigma^2 ) , & \text{for $t\leq \sigma^2 / R$,} \\
   d \exp (-t / 4R ) , & \text{for $t\geq \sigma^2 / R$.}
  \end{cases}
 \end{equation}
\end{theorem}

In our case the matrices $\vc{X}_k$ are given by the matrices $\widehat{\vc{Y}}^{(i,a)}$ being defined in \eqref{eqn:block_struct}.
Clearly these matrices are self-adjoint and, moreover $\E \widehat{\vc{Y}}^{(i,a)} = 0$.
In order to provide for all further inequalities in \eqref{eqn:tropp_moments} we first calculate powers of the block matrices $\vc{\Lambda}_{\tau}^{(i,a)}$ in \eqref{eqn:block_struct}:
\begin{equation}\label{eqn:Lambda_powers}
 \big[\vc{\Lambda}_{\tau}^{(i,a)}\big]^p = \big\langle \vc{v}_{\tau}^{(i,a)} , \vc{v}_{\tau}^{(i,a)} \big\rangle^{p-2} \big[\vc{\Lambda}_{\tau}^{(i,a)}\big]^2
 = \left( \frac{|S_{[\beta]}^\tau|}{N_T N_t} \right)^{p-2} \big[\vc{\Lambda}_{\tau}^{(i,a)}\big]^2 .
\end{equation}
The random variables $\chi_{(i,a)}$ in \eqref{eqn:block_struct} are subexponential with 
\[
 \P ( |\chi_{(i,a)}| \geq t ) = \P \big( \big| |s_{(i,a)}|^2 - 1 \big| \geq t \big) \leq \P ( |s_{(i,a)}|^2 \geq t-1 )  \leq e^{-(t-1)} = e e^{-t},
\]
where the second inequality follows from Lemma~\ref{lem:complexGaussian}.
Using basically Fubini's Theorem and a change of variables (see, e.g., \cite[Prop. 7.1]{Foucart2013}) and, furthermore, plugging in the tail estimate from above, we obtain for $p\geq 2$,
\begin{equation}\label{eqn:chi_powers}
 \E |\chi_{(i,a)}|^p = p \int_0^\infty t^{p-1} \P ( |\chi_{(i,a)}| \geq t) ~ dt \leq ep \int_0^\infty t^{p-1} e^{-t} ~ dt = 2e \frac{p!}{2} .
\end{equation}
For the last equality we used that the appearing integral is equal to $\Gamma (p) = (p-1)!$, where $\Gamma$ denotes the well-known Gamma function.
Finally, by combining \eqref{eqn:Lambda_powers} with \eqref{eqn:chi_powers} and recalling \eqref{eqn:block_struct} we observe that
\newcommand{\randYia}{\widehat{\vc{Y}}_{\hspace{-.3em}(i,a)}}
\[
 \E \big[\randYia\big]^p \preccurlyeq \frac{p!}{2} \left( \frac{\max_{\tau \in [N_t]} |S_{[\beta]}^\tau|}{N_T N_t} \right)^{p-2} \big[\sqrt{2e}\vc{Y}^{(i,a),(i,a)}\big]^2 .
\]
With the matrices $\sqrt{2e}\vc{Y}^{(i,a),(i,a)}$ in place of the matrices $\vc{A}_k$, the quantity \eqref{eqn:tropp_variance} can be calculated as
\[
 \sigma^2 = \bigg\| \sum_{i=1}^{N_T} \sum_{a=1}^{N_t} \big[\sqrt{2e}\vc{Y}^{(i,a),(i,a)}\big]^2 \bigg\|_{2\to 2} = 2e \bigg\| \sum_{i=1}^{N_T} \sum_{a=1}^{N_t} \vc{Y}^{(i,a),(i,a)} \vc{Y}^{(i,a),(i,a)} \bigg\|_{2\to 2} = 2e \frac{\max_{\tau \in [N_t]} |S_{[\beta]}^\tau|}{N_T N_t} ,
\]
where we used that
\[
 \bigg[ \sum_{i=1}^{N_T} \sum_{a=1}^{N_t} \vc{Y}^{(i,a),(i,a)} \vc{Y}^{(i,a),(i,a)} \bigg]_{\varTheta ,\varTheta^\prime} = \delta_{\varTheta,\varTheta^\prime} \frac{|S_{[\beta]}^{\tau}|}{N_T N_t} .
\]
(The proof is analog to the identities in \eqref{eqn:matrix_sums}, see the appendix.)
A direct application of Theorem \ref{thm:tropp}, where the first bound in \eqref{eqn:tropp_tail_bounds} applies, yields
\begin{equation}\label{eqn:YY_gleich_tail2}
 \P ( \opnorm{\vc{\YY}_{=} - \ID} \geq \delta / 2 ) \leq |S_{[\beta]}| \exp \left( - \frac{(\delta /2)^2 N_T N_t}{4\cdot 2e \max_{\tau \in [N_t]} |S_{[\beta]}^\tau|} \right) \leq |S_{[\beta]}| \exp \left( - \frac{\delta^2 N_T N_t}{32e |S_{[\beta]}|} \right) .
\end{equation}

\subsubsection*{A tail bound for $\opnorm{\vc{\YY}_{\neq}}$}
In order to derive a tail estimate for the quantity $\opnorm{\vc{\YY}_{\neq}}$, we examine the moments
$\E \opnorm{\vc{\YY}_{\neq}}^{2m}$, $m\in\N$.
We use decoupling (see, e.g., \cite[Thm. 8.11]{Foucart2013}) in order to obtain
\begin{equation}\label{eqn:first_summand}
 \E \opnorm{\vc{\YY}_{\neq}}^{2m} = \E \bigg\| \mathop{\sum_{i,j=1}^{N_T} \sum_{a,b=1}^{N_t}}_{(i,a) \neq (j,b)} \overline{{s}_{(i,a)}} {s}_{(j,b)} \vc{Y}^{(i,a),(j,b)} \bigg\|^{2m}_{2 \to 2} \leq 4^{2m} \E \bigg\| \sum_{i,j=1}^{N_T} \sum_{a,b=1}^{N_t} \xi_{(i,a)} \xi'_{(j,b)} \vc{Y}^{(i,a),(j,b)} \bigg\|_{2\to 2}^{2m},
\end{equation}
where the $\xi_{(i,a)}, \xi'_{(j,b)}$, $(i,a),(j,b) \in [N_T] \times [N_t]$, 
are independent standard complex Gaussian random variables.
The term on the right-hand side can be estimated by means of the following Khintchine-type inequality for a homogeneous matrix-valued chaos of order two taken from \cite{ra10}\footnote{
Note that due to the comment after the formulation of the original result in \cite{ra10} the extra factor $2^{1/2n}$ appearing right-hand side of the assertion in \cite[Thm. 6.22]{ra10} can be removed.
In the original version of \cite[Thm. 6.22]{ra10}, the quantity $\|\widetilde{\vc{F}}\|_{S_{2m}}^{2m}$ was missing; this has been corrected in a new version which can be obtained from the personal website of HR.
}.
Although the original result \cite[Thm. 6.22]{ra10} is formulated for independent Rademacher random variables, a careful observation of the proof reveals that it is also true for standard complex Gaussians.

\begin{lemma}\label{lem:khintchine}
Let $\vc{B}_{k, \ell}$, $k, \ell = 1, \ldots M$, be complex matrices of the same dimension and let $\xi_k , \xi_\ell^\prime$, $k , \ell \in [M]$, independent standard Gaussian random variables.
Then, for $m \in\N$,
\begin{gather}
 \E \bigg\| \sum_{k, \ell = 1}^M \xi_k \xi'_{\ell} \vc{B}_{k, \ell} \bigg\|_{S_{2m}}^{2m} \notag   \\
 \leq \bigg( \frac{(2m)!}{2^m m!} \bigg)^2 \max \bigg\{ \bigg\| \bigg( \sum_{k, \ell = 1}^M \vc{B}_{k, \ell} \vc{B}_{k, \ell}^* \bigg)^{1/2} \bigg\|_{S_{2m}}^{2m} , \bigg\| \bigg( \sum_{k, \ell = 1}^M \vc{B}_{k, \ell}^* \vc{B}_{k, \ell} \bigg)^{1/2} \bigg\|_{S_{2m}}^{2m} , \|\vc{F}\|_{S_{2m}}^{2m} , \|\widetilde{\vc{F}}\|_{S_{2m}}^{2m} \bigg\}, \label{eqn:khintchine1}
\end{gather}
where $\vc{F}, \widetilde{\vc{F}}$ are the block matrices $\vc{F} = (\vc{B}_{k, \ell})_{k, \ell = 1}^M$ and
$\widetilde{\vc{F}} = (\vc{B}_{k, \ell}^*)_{k, \ell = 1}^M$, respectively.
\end{lemma}

\goodbreak
In our setting, the matrices $\vc{B}_{k,\ell}$ are given by the matrices $\vc{Y}^{(i,a),(j,b)}$.
Therefore, $\vc{F}$ and $\widetilde{\vc{F}}$ appearing in the assertion of Lemma~\ref{lem:khintchine} are block matrices consisting of $|S_{[\beta]}| \times |S_{[\beta]}|$ blocks.
In the following we calculate the Schatten $2m$-norms of the matrix $\vc{F}$.
To this end we compute the eigenvalues of $\vc{F}^* \vc{F}$.
Like the matrix $\vc{F}$ itself also the product $\vc{F}^* \vc{F}$ consists of $|S_{[\beta]}|$-dimensional blocks.
A straightforward calculation using the definition of the matrices reveals (see the end of Appendix \ref{apdx:section3} for the details) that the block at position $(i,a),(j,b)$ of this matrix product is entrywise given by
\[
 [\vc{F}^* \vc{F}]^{(i,a),(j,b)}_{\varTheta ,\varTheta^\prime}
 = \delta^{\sim_{N_t}}_{\tau^\prime + b , \tau + a} N_T^{-1} N_t^{-2} |S_{[\beta]}| \ee{\frac{f'-f}{N_t} (\tau +a-1)} \ee{d_T \Delta_\beta [\beta^\prime (j-1)-\beta (i-1)]} .
\]
Due to the appearance of the factor $\delta^{\sim_{N_t}}_{\tau^\prime + b , \tau + a}$ the matrix $\vc{F}^* \vc{F}$ can be transformed into a block diagonal matrix by merely rearranging the ordering of the rows/columns.
Hence, we assume without loss of generality that $\vc{F}^* \vc{F}$ is block diagonal, with the blocks $\boldsymbol{\Lambda}_1 , \boldsymbol{\Lambda}_2 , \ldots , \boldsymbol{\Lambda}_{N_t}$ on the diagonal, where each block $\boldsymbol{\Lambda}_\kappa$ is a rank one matrix of the form
\[
 \vc{\Lambda}_\kappa = \frac{|S_{[\beta]}|}{N_T N_t^{2}} \vc{w}_{\kappa} \vc{w}_{\kappa}^* ,
\]
and where for each $\kappa\in [N_t]$ the vector $\vc{w}_\kappa \in \C^{N_T |S_{[\beta]}|}$ is entrywise given by
\[
 [\vc{w}_{\kappa}]_{i,\varTheta} = \ee{\frac{f}{N_t}(\kappa -1)} \ee{d_T \Delta_\beta \beta (i-1)} , \qquad i\in [N_T],\, \varTheta \in S_{[\beta]} .
\]
Hence, $\| \vc{w}_\kappa \|_2^2 = N_T |S_{[\beta]}|$ and, therefore, $\vc{w}_{\kappa} \vc{w}_{\kappa}^*$ has exactly one nonzero eigenvalue, equal to $N_T |S_{[\beta]}|$.
This implies that also the block $\vc{\Lambda}_\kappa$ has exactly one nonzero eigenvalue which is equal to $|S_{[\beta]}|^2 /N_t^2$ and, thus, $\| \vc{\Lambda}_\kappa \|_{S_{m}} = |S_{[\beta]}|^2 /N_t^2$.
With this we can calculate,
\begin{equation}\label{eqn:max1}
 \| \vc{F} \|_{S_{2m}}^{2m} = \| \vc{F}^* \vc{F} \|_{S_{m}}^{m} = \sum_{\kappa \in [N_t]} \| \vc{\Lambda}_\kappa \|_{S_{m}}^{m} = N_t \frac{|S_{[ \beta ]}|^{2m}}{N_t^{2m}} = |S_{[ \beta ]}|  \frac{|S_{[ \beta ]}|^{2m-1}}{N_t^{2m-1}} \leq |S_{[ \beta ]}|  \frac{|S_{[ \beta ]}|^{m}}{N_t^{m}} ,
\end{equation}
where we used that, without loss of generality, $|S_{[ \beta ]}| \leq N_t$.
By using the same arguments for the matrix $\widetilde{\vc{F}}$ it is straightforward to verify that also
\begin{equation}\label{eqn:max2}
 \| \widetilde{\vc{F}} \|_{S_{2m}}^{2m} \leq |S_{[ \beta ]}|  \frac{|S_{[ \beta ]}|^{m}}{N_t^{m}}
\end{equation}
holds true.

\goodbreak
We proceed with estimating the first two quantities in the maximum expression in \eqref{eqn:khintchine1}.
As we calculate in Appendix \ref{apdx:section3} (see \eqref{eqn:matrix_sums}),
\[
 \sum_{i,j=1}^{N_T} \sum_{a,b=1}^{N_t} \vc{Y}^{(i,a),(j,b)} [ \vc{Y}^{(i,a),(j,b)} ]^* = \sum_{i,j=1}^{N_T} \sum_{a,b=1}^{N_t} [ \vc{Y}^{(i,a),(j,b)} ]^* \vc{Y}^{(i,a),(j,b)} = \frac{|S_{[\beta]}|}{N_t} \ID .
\]
Therefore, the corresponding Schatten $2m$-norm in \eqref{eqn:khintchine1} can be calculated as
\begin{equation*}
 \bigg\| \bigg[ \sum_{i,j=1}^{N_T} \sum_{a,b=1}^{N_t} \vc{Y}^{(i,a),(j,b)} [ \vc{Y}^{(i,a),(j,b)} ]^* \bigg]^{1/2} \bigg\|_{S_{2m}}^{2m} = \bigg\| \bigg[ \sum_{i,j=1}^{N_T} \sum_{a,b=1}^{N_t} [ \vc{Y}^{(i,a),(j,b)} ]^* \vc{Y}^{(i,a),(j,b)} \bigg]^{1/2} \bigg\|_{S_{2m}}^{2m} = |S_{[ \beta ]}| \frac{|S_{[ \beta ]}|^{m}}{N_t^m} .
\end{equation*}
Now a direct application of the Khintchine-type inequality (see Lemma \ref{lem:khintchine} above) to the moments in \eqref{eqn:first_summand}, using \eqref{eqn:max1}, \eqref{eqn:max2}, and the latter equalities implies
\begin{equation}\label{eqn:first_moments}
 \E \|\vc{\YY}_{\neq}\|_{2\to 2}^{2m} \leq 4^{2m} \bigg( \frac{(2m)!}{2^m m!} \bigg)^{2} |S_{[ \beta ]}| \frac{|S_{[ \beta ]}|^m}{N_t^m} = \bigg( \frac{(2m)!}{2^m m!} \bigg)^{2} |S_{[ \beta ]}| \left( \frac{4 \sqrt{|S_{[ \beta ]}|}}{\sqrt{N_t}} \right)^{2m} .
\end{equation}
Using Stirling's formula, one can easily verify that
\[
 \frac{(2m)!}{2^m m!} \leq \sqrt{2} (2/e)^m m^m ,
\]
Therefore, the moments in \eqref{eqn:first_moments} can further be estimated as
\begin{equation}\label{eqn:first_moments2}
 \E \|\vc{\YY}_{\neq}\|_{2\to 2}^{2m} \leq 2 |S_{[ \beta ]}| \left( \frac{8 \sqrt{|S_{[ \beta ]}|}}{e \sqrt{N_t}} \right)^{2m} m^{2m} ,
\end{equation}
which holds true for $m\in\N \cup \{0\}$.
H\"older's inequality implies
\[
\E |Z|^{2m + 2 \theta} = \E \big[ |Z|^{(1-\theta) 2m} |Z|^{\theta (2m+2)} \big] \leq (\E |Z|^{2m} )^{1-\theta} (\E |Z|^{2m + 2})^\theta,
\]
for any random variable $Z$, and each $\theta \in [0,1]$.
Combining this with \eqref{eqn:first_moments2} gives
\[
 \E \|\vc{Y}_{\neq}\|_{2\to 2}^{2m+2\theta} \leq 2 |S_{[ \beta ]}| \left( \frac{8 \sqrt{|S_{[ \beta ]}|}}{e \sqrt{N_t}} \right)^{2m+2\theta} m^{2m(1-\theta)} (m+1)^{(2m+2)\theta} .
\]
For $m\in\N$ we estimate the last terms as
\begin{equation}\label{eqn:agm}
 m^{2m(1-\theta)} (m+1)^{(2m+2)\theta} = \big( m^{1-\theta} (m+1)^\theta \big)^{2m+2\theta} \bigg( \frac{m+1}{m} \bigg)^{2\theta (1-\theta)} \leq \sqrt{2} (m+\theta)^{2m+2\theta},
\end{equation}
where we used the inequality between the (weighted) arithmetic and geometric mean and, furthermore, used the fact that $(m+1)/m \leq 2$ and $\theta (1-\theta) \leq 1/4$.
Since
\[
 \inf \{ x^x \midcol x \geq 0 \} = e^{-1/e} \geq 3/5 ,
\]
we can replace $\sqrt{2}$ on the right-hand side in \eqref{eqn:agm} by $3$ leaving us with an inequality which is now also valid for $m=0$.
Altogether it holds for all real $p = 2m + 2\theta > 0$,
\[
 \big( \E \|\vc{\YY}_{\neq}\|_{2\to 2}^{p} \big)^{1/p} \leq \frac{4\sqrt{|S_{[ \beta ]}|}}{e\sqrt{N_t}} \big( 6 |S_{[ \beta ]}| \big)^{1/p} p.
\]
In order to obtain a tail estimate for the random variable $\|\vc{\YY}_{\neq}\|_{2\to 2}$ it only remains to apply the following lemma which is a direct implication of Markov's inequality.
For a proof see, e.g., \cite{ra10,Foucart2013}.

\begin{lemma}\label{lem:p_moments_tail_est}
 Let $Z$ be a random variable and suppose there exist constants $\alpha , \beta , \gamma , p_0 > 0$ such that
 \[
  ( \E |Z|^p )^{1/p} \leq \alpha \beta^{1/p} p^{1 / \gamma},
 \]
 for all $p \geq p_0$.
 Then it holds for all $u \geq p_0^{1/\gamma}$,
 \[
  \P ( |Z| \geq e^{1 / \gamma} \alpha u ) \leq \beta e^{-u^{\gamma} / \gamma} .
 \]
\end{lemma}

By applying Lemma \ref{lem:p_moments_tail_est} we obtain
\begin{equation}\label{eqn:YY_ungleich_tail}
 \P \big( \opnorm{\vc{\YY}_{\neq}} \geq \delta / 2 \big) \leq 6 |S_{[\beta]}| \exp \bigg( - \frac{\delta \sqrt{N_t}}{8\sqrt{|S_{[\beta]}|}} \bigg) .
\end{equation}

\subsubsection*{Conclusion}
Finally, due to the definition of the matrices $\vc{\YY}_{\neq}$, $\vc{\YY}_{=}$ in \eqref{eqn:tail_triangle}, an application of the union bound and plugging in the tail bounds \eqref{eqn:YY_gleich_tail2}, \eqref{eqn:YY_ungleich_tail} (noting
that \eqref{eqn:YY_gleich_tail2} is dominated by \eqref{eqn:YY_ungleich_tail})
implies the assertion of Proposition \ref{prop:opnorm_bound}.\qed

\section{Numerical experiments}\label{sec:numerics}
We conduct numerical experiments in order to test the effective empirical probabilities of perfect reconstruction of support sets using the LASSO functional, as considered in Theorem \ref{thm:nonuniform_result}.
Due to runtime limitations we only perform simulations for the Doppler-free scenario.
As pointed out in in \mbox{Section \ref{sec:doppler_free_case}}, this case is contained in our analysis as a special case.
This means in particular that the same influence of the parameter $\eta$ on the recovery properties can be expected.
This is indeed supported by the results of the numerical experiments described next.

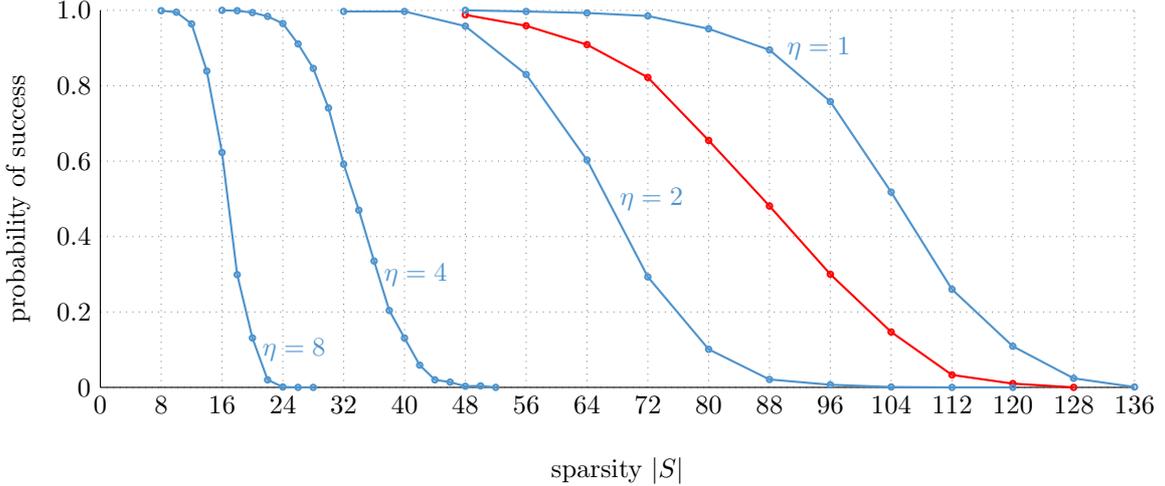
\begin{figure}[t!]
\begin{center}
\begin{tikzpicture}[xscale=0.1,yscale=1]

  \def\xmin{0}
  \def\xmax{136}
  \def\ymin{0}
  \def\ymax{5}

  \draw[draw=Azure4,style=dotted, ystep=1, xstep=8] (\xmin,\ymin) grid
  (\xmax,\ymax);

  
  \draw (0,0) -- coordinate (y axis mid) (0,5);
  \draw (0,0) -- coordinate (x axis mid) (136,0);
  
  \node[below=0.8cm] at (x axis mid) {sparsity $|S|$};
  \node[rotate=90, above=0.8cm] at (y axis mid) {probability of success};

  \foreach \x in {0,8,...,136}
    \node at (\x, \ymin) [below] {\x};
  \foreach \y in {0,0.2,0.4,0.6,0.8,1.0}
    \node at (\xmin,5*\y) [left] {\y};

    
  \draw[color=red,thick] plot[mark=o, mark options={xscale=5,yscale=0.5}] file {B_free_scaled.dat};

  
  \draw[color=SteelBlue3,thick] plot[mark=o, mark options={xscale=5,yscale=0.5}]
        file {B1_scaled.dat};
  
\node[color=SteelBlue3,thick] at (89,4.5) [right] {$\eta=1$};
  
  
  \draw[color=SteelBlue3,thick] plot[mark=o, mark options={xscale=5,yscale=0.5}] file {B2_scaled.dat};
   \node[color=SteelBlue3,thick] at (67,2.5) [right] {$\eta=2$};
   \draw[color=SteelBlue3,thick] plot[mark=o, mark options={xscale=5,yscale=0.5}] file {B4_scaled.dat};
   \node[color=SteelBlue3,thick] at (36,1.5) [right] {$\eta=4$};
   
   \draw[color=SteelBlue3,thick] plot[mark=o, mark options={xscale=5,yscale=0.5}] file {B8_scaled.dat};
   \node[color=SteelBlue3,thick] at (20,0.5) [right] {$\eta=8$};


\end{tikzpicture}
\end{center}
\caption{Probability of perfect reconstruction of support sets for varying sparsity levels $|S|$.
Blue plots: Randomly chosen support sets having constant balancedness parameter $\eta$ as indicated.
Red plot: Randomly chosen support sets with no restriction on the balancedness parameter.}
\label{B_plots}
\end{figure}

We consider the case where $N_T = N_R = 8$ and $N_t = 64$, yielding a problem dimension of \mbox{$N=N_T N_R N_t = 4096$} and a total number of $m=N_R N_t = 512$ measurements.
We run simulations with varying choices for $(|S|,\eta )$, where $|S|$ is the number of targets and $\eta$ defines the actual balancedness of the considered support sets (see Definition \ref{def:balanced}).
For each such pair $(|S|,\eta )$ we try to reconstruct $1000$ independent random target scenes $\vc{x}$ from measurements $\vc{y} = \vc{A}\vc{x} + \vc{n}$ by minimizing the LASSO functional \eqref{eqn:lasso} with $\lambda = 2 \sigma \sqrt{2 N_T N_R N_t \log (N)}$, and where, in each run, 
\begin{itemize}
 \item the signal vectors $\vc{s}_i$ (generating the measurement matrix $\vc{A}$) and also the noise vector ${\vc{n}}$ are drawn randomly from standard complex Gaussian distribution, i.e., we choose $\sigma^2 = 1$ as the variance of the noise,
 \item the support set $S$ of $\vc{x}$ is drawn uniformly at random from all possible support sets having balancedness parameter $\eta$,
 \item the phases of the nonzero entries of $\vc{x}$ are drawn uniformly at random from $[0 , 2\pi )$,
 \item and the amplitudes $| x_\varTheta |$, $\varTheta \in S$, are set to the threshold level defined in \eqref{eqn:thres}.
\end{itemize}

We use the Matlab library TFOCS \cite{Becker2011} for the minimization of the LASSO functional \eqref{eqn:lasso}.
Since TFOCS uses an iterative solver and, hence, in general only approximations to the desired minimizer can be computed, 
and we consider a target scene $\vc{x}$ to be successfully recovered if $\|\vc{x} - \vc{x}^\#\|_\infty \leq \TT$ for an appropriate
thresholding $\TT > 0$.

In Figure~\ref{B_plots} we depict the results of our simulations.
It becomes evident that small values of the parameter $\eta$ have a positive impact on the probability of perfect reconstruction of the support set $S$.
Moreover, as predicted by formula \eqref{eqn:nonuniform_scaling}, the dependence of the number of maximal reconstructable targets $|S|$ (for a fixed number of measurements $N_R N_t$) on the parameter $\eta$ is linear.

\section*{Acknowledgment}
The authors acknowledge funding from the European Research Council through the Starting Grant \mbox{StG 258926}.

\section*{Appendix}
\appendix

\section{Orthogonality of the matrices \texorpdfstring{$\x{\varTheta}$}{Xtheta}}\label{apdx:section2}
Due to the definition in \eqref{eqn:XX}, for $(i,j) \in [N_R] \times [N_T]$, the $(i,j)$th $N_t \times N_t$ block of $\x{\varTheta}$ is given by
\[
\x{\varTheta}^{[i,j]} = \ee{d_R \beta \Delta_\beta (i-1)} \ee{d_T \beta \Delta_\beta (j-1)} \vc{M}_{f} \vc{T}_{\tau} .
\]

\begin{lemma}
 The set of matrices $\left\{ \frac{1}{\sqrt{N_t N_R N_T}} \x{\varTheta} ,\, \varTheta \in \GG \right\}$ forms an orthonormal basis. 
\end{lemma}
\begin{proof}
We calculate the inner product $\langle \x{\varTheta'} , \x{\varTheta} \rangle = \text{Tr} ( \x{\varTheta^\prime}^* \x{\varTheta} )$.
To this end we calculate the $(j^\prime,j)$th block of the product $\x{\varTheta^\prime}^* \x{\varTheta}$ (recall, $\x{\varTheta^\prime}^* \x{\varTheta}$ is a $N_T \times N_T$ block matrix consisting of $N_t \times N_t$ blocks),
\begin{align}
 [ \x{\varTheta'}^* \x{\varTheta} ]^{[j^\prime,j]}
 &= \sum_{k=1}^{N_R} [ \x{\varTheta'}^* ]^{[j^\prime,k]} \x{\varTheta}^{[k,j]} = \sum_{k=1}^{N_R} [ \x{\varTheta'}^{[k,j^\prime]} ]^* \x{\varTheta}^{[k,j]} \notag \\
 &= \sum_{k=1}^{N_R} \overline{\ee{d_R \beta^\prime \Delta_\beta (k -1)} \ee{d_T \beta^\prime \Delta_\beta (j^\prime -1)}} \ee{d_R \beta \Delta_\beta (k-1)} \ee{d_T \beta \Delta_\beta (j-1)} \vc{T}_{\tau^\prime}^* \vc{M}_{f^\prime}^* \vc{M}_{f} \vc{T}_{\tau}  \notag \\
 &= \nee{d_T \beta^\prime \Delta_\beta (j^\prime-1)} \ee{d_T \beta \Delta_\beta (j - 1)} \ee{\frac{(f-f')}{N_t} \tau^\prime} \sum_{k=1}^{N_R} \ee{d_R (\beta - \beta^\prime ) \Delta_\beta (k-1)} \vc{M}_{f-f'} \vc{T}_{\tau - \tau'} \label{eqn:block_product_XX}
\end{align}
For the last equality we used that
\[
 \vc{T}_{\tau^\prime}^* \vc{M}_{f^\prime}^* \vc{M}_{f} \vc{T}_{\tau} = \ee{\frac{(f-f')}{N_t} \tau^\prime} \vc{M}_{f-f'} \vc{T}_{\tau-\tau'} ,
\]
which follows directly from the definitions of the operators $\vc{M}_{f}$ and $\vc{T}_{\tau}$ (see \eqref{eqn:signal_ops_def}).
Due to \eqref{eqn:block_product_XX}, the Frobenius inner product between two matrices is given as
\begin{align*}
 \langle \x{\varTheta'} , \x{\varTheta} \rangle
 &= \sum_{j=1}^{N_T} \text{Tr} ( [ \x{\varTheta'}^* \x{\varTheta} ]^{[j,j]} )  \\
 &=  \bigg( \sum_{j=1}^{N_T} \ee{d_T (\beta - \beta') \Delta_\beta (j-1)} \bigg) \ee{\frac{(f-f')}{N_t} \tau^\prime} \sum_{k=1}^{N_R} \ee{d_R (\beta-\beta') \Delta_\beta (k-1)} \text{Tr} ( \vc{M}_{f-f'} \vc{T}_{\tau - \tau'} ) .
\end{align*}
Since $\vc{M}_{f-f'}$ is a diagonal matrix, the trace of the product $\vc{M}_{f-f'} \vc{T}_{\tau-\tau'}$ can only be nonzero if at least one of the diagonal entries of the matrix $\vc{T}_{\tau-\tau'}$ is nonzero, i.e., if $\tau = \tau'$ so that $\vc{T}_{\tau-\tau'} = \ID$.
This means that
\[
 \text{Tr} ( \vc{M}_{f-f'} \vc{T}_{\tau-\tau'} ) = \text{Tr} ( \vc{M}_{f-f'} ) = \sum_{k=1}^{N_t} \ee{\frac{f-f'}{N_t} (k-1)} =
 \begin{cases}
 N_t  &  \text{if $f = f'$,}  \\
 0    &  \text{otherwise.}
 \end{cases}
\]
Recalling the formula for $\langle \x{\varTheta'} , \x{\varTheta} \rangle$ from above implies that for this inner product to be nonzero it necessarily has to hold that $\varTheta' = \varTheta$.
Indeed, this follows from the appearance of the factor
$
 \sum_{j=1}^{N_T} \ee{d_T (\beta - \beta') \Delta_\beta (j-1)}
$
which (recalling that $d_T = 1/2$ and $\Delta_\beta = 2/ N_T N_R$, see \eqref{eqn:dT_dR}, \eqref{eqn:stepsizes}) is only nonzero (and equal to $N_T$) if $\beta^\prime = \beta$.
Finally, we can conclude
\[
 \langle \x{\varTheta'} , \x{\varTheta} \rangle = \delta_{\varTheta',\varTheta} N_T N_R N_t .
\]
The normalization yields the result.
\end{proof}

\section{Proof of Lemma \ref{lem:cov_number_bounds}}\label{apdx:proof_cov_numbers}
For small $u$ the first term on the right-hand side of \eqref{eqn:covering_numbers_bound} 
can be obtained by a volumetric argument.
To this end, let for $S\subset \{1, \ldots , N\}$, $B_S \subset \C^N$ denote the set of all vectors $\vc{x}$ with $\|\vc{x}\|_2 \leq 1$ and support in $S$.
Introducing $|\hspace{-.1em} \| \vc{x} |\hspace{-.1em}\| := \opnorm{\vxTilde}$ we find using \eqref{eqn:opnorm_by_euclid} and a volumetric estimate, see e.g.~\cite[Proposition~C.3]{Foucart2013},
\[
 \NN( B_S, |\hspace{-.1em}\| \cdot |\hspace{-.1em}\|, u ) \leq \NN \big( B_S, \sqrt{s/N_t} \|\cdot\|_2, u \big) \leq \bigg(1+ 2 \frac{\sqrt{s/N_t}}{u} \bigg)^{2s} \leq \bigg(3 \frac{\sqrt{s/N_t}}{u} \bigg)^{2s},
\]
where for the last inequality we used the assumption $u \leq \sqrt{s/N_t}$.
Since $\AA$ is the union of all sets $B_S$ with $S \subset [N]$, and there are $\genfrac(){0pt}{}{N}{s} \leq (eN / s)^s$ possible choices for $S$ it holds
\[
 \NN( \AA, \opnorm{\cdot}, u ) \leq (eN / s)^s (3 \sqrt{s/N_t} / {u})^{2s},
\]
which implies the first bound in \eqref{eqn:covering_numbers_bound}, namely
\[
\log \NN(\AA, \opnorm{\cdot}, u) \leq 2s \left( \log (eN / s) + \log \bigg(\frac{2\sqrt{s}}{u\sqrt{N_t}} \bigg) \right) \lesssim s \log \bigg( \frac{N}{u^2 N_t} \bigg) .
\]
For the second bound from the assertion we exploit the fact that
\[
 \{ \vc{x} \in \C^{N} \midcol \text{$\vc{x}$ $s$-sparse, $\|\vc{x}\|_2 \leq 1$}\} \subset \sqrt{2s} \conv \bigcup_{\varTheta \in \GG} \{ \boldsymbol{e}_{\varTheta} ,\, \imath \boldsymbol{e}_{\varTheta} ,\, -\boldsymbol{e}_{\varTheta} ,\, -\imath \boldsymbol{e}_{\varTheta} \} =: \widetilde D_s
\]
and, hence, the set $\AA$ from \eqref{eqn:AA_set} is contained in the set $\{ \vxTilde \midcol \vc{x} \in \widetilde D_s \}$.
The following is a version of \textit{Maurey's lemma}.
For a proof see, e.g., \cite{krmera14}.

\begin{lemma}\label{lem:maurey}
There exists an absolute constant $c$ for which the following holds.
Let $X$ be a normed space, consider a finite set $\UU \subset X$ of cardinality $N$, and assume that for every $L\in \N$ and
$(\vc{u}_1 , \ldots , \vc{u}_L ) \in \UU^L$, $\E \| \sum_{j=1}^L {\epsilon}_j \vc{u}_j \|_X \leq A \sqrt{L}$, where
$(\epsilon_1 , \ldots , \epsilon_L )$ denotes a Rademacher vector.
Then for every $u>0$,
\[
 \log \NN( \conv ( \UU ) , \|\cdot \|_X , u) \leq c (A/u)^2 \log (N) .
\]
\end{lemma}

In order to apply Lemma \ref{lem:maurey}, we need to estimate the quantity $\E_{\vc{\epsilon}} \opnorm{\sum_{k=1}^L \epsilon_k \vxTildeArg{\vc{u}_k}}$, where $(\vc{u}_1 , \ldots , \vc{u}_L )$ is a sequence of extreme points in $\widetilde D_s$ and $\vc{\epsilon} = (\epsilon_1 , \ldots , \epsilon_L )$ is a Rademacher vector.
The noncommutative Khintchine inequality \cite{bu01,ra10} -- originally due to Lust-Piquard \cite{lu86-1,lupi91} --
or more modern estimates based on moment generating function bounds \cite{tr12},
see also \cite[Problem~8.6(d)]{Foucart2013}, 
\begin{equation}\label{eqn:lust_picard}
 \E_{\vc{\epsilon}} \opnorm{\sum_{k=1}^L \epsilon_k \vxTildeArg{\vc{u}_k}} \lesssim \sqrt{\log ( N_{\text{max}} )} \max \bigg\{ \opnorm{\sum_{k=1}^L \vxTildeArg{\vc{u}_k} \vxTildeArg{\vc{u}_k}^*}, \opnorm{\sum_{k=1}^L \vxTildeArg{\vc{u}_k}^* \vxTildeArg{\vc{u}_k} } \bigg\}^{{1}/{2}} ,
\end{equation}
where $N_{\text{max}}$ stands for the maximum of the dimensions of the matrices $\vxTildeArg{\vc{u}_k} \vxTildeArg{\vc{u}_k}^*$ and $\vxTildeArg{\vc{u}_k}^* \vxTildeArg{\vc{u}_k}$, and can be estimated by $\max \{ N_R N_t , N_T N_t \} \leq N$.
Using the estimate \eqref{eqn:opnorm_by_euclid} for $\opnorm{\vxTildeArg{\vc{u}_k}}$,
\[
 \opnorm{\vxTildeArg{\vc{u}_k} \vxTildeArg{\vc{u}_k}^*} = \opnorm{\vxTildeArg{\vc{u}_k}^* \vxTildeArg{\vc{u}_k}} = \opnorm{\vxTildeArg{\vc{u}_k}}^2 \leq \frac{1}{N_t} \| \vc{u}_k \|_1^2 = \frac{2s}{N_t}.
\]
An application of the triangle inequality yields, using the Khintchine inequality \eqref{eqn:lust_picard},
\[
 \E_{\vc{\epsilon}} \opnorm{\sum_{k=1}^L \epsilon_k \vxTildeArg{\vc{u}_k}} \lesssim \sqrt{\log (N)}  \sqrt{\frac{2s}{N_t}} \sqrt{L}
\]
Finally, we can apply Lemma \ref{lem:maurey} yielding
\[
\log \NN(\AA, \opnorm{\cdot}, u) \lesssim \frac{s}{u^2 N_t} \log^2 (N) .
\]
This establishes the second bound in \eqref{eqn:covering_numbers_bound}.\qed

\section{Basic calculations for Proposition~\ref{prop:opnorm_bound}}\label{apdx:section3}
The proof of Proposition \ref{prop:opnorm_bound} is based on the fact that 
\begin{equation}\label{eqn:AA_repr_2}
 \widetilde{\vc{A}}_{S_{[\beta]}}^* \widetilde{\vc{A}}_{S_{[\beta]}} = \sum_{i,j = 1}^{N_T} \sum_{a,b = 1}^{N_t} \overline{{s}_{(i,a)}} {s}_{(j,b)} \vc{Y}^{(i,a),(j,b)} ,
\end{equation}
where we write $s_{(i,a)}$ for the $a$-th entry of the signal vector $\vc{s}_i$ and where for $\varTheta , \varTheta^\prime \in S_{[\beta]}$ the corresponding entry of a given matrix $\vc{Y}^{(i,a),(j,b)}$ is given by
\begin{equation}\label{eqn:def_B_matrix_in_sum_2}
 [\vc{Y}^{(i,a),(j,b)}]_{\varTheta,\varTheta'} = \delta^{\sim_{N_t}}_{a-b , \tau^\prime -\tau} (N_T N_t)^{-1} \ee{\frac{f'-f}{N_t} (\tau +a-1)} \ee{d_T \Delta_\beta [\beta^\prime (j-1)-\beta (i-1)]} ,
\end{equation}
To see this we recall that, according to \eqref{eqn:correlation_formula},
the inner products $\langle \vc{A}_{\varTheta} , \vc{A}_{\varTheta^\prime} \rangle$ satisfy
\begin{align*}
 [ \widetilde{\vc{A}}_{S_{[\beta]}}^* \widetilde{\vc{A}}_{S_{[\beta]}} ]_{\varTheta , \varTheta^\prime}
  &= \langle \widetilde{\vc{A}}_{\varTheta} , \widetilde{\vc{A}}_{\varTheta^\prime} \rangle = (N_T N_R N_t)^{-1} \langle \vc{A}_{\varTheta} , \vc{A}_{\varTheta^\prime} \rangle  \\
  &= (N_T N_t)^{-1} \sum_{i,j=1}^{N_T} \ee{d_T \Delta_\beta [\beta^\prime (j-1) - \beta (i-1)]} \left\langle \vc{M}_{f} \vc{T}_{\tau} \vc{s}_i , \vc{M}_{f^\prime} \vc{T}_{\tau^\prime} \vc{s}_j \right\rangle ,
\end{align*}
where we used that both $\varTheta , \varTheta^\prime \in S_{[\beta]}$.
Recalling the definitions of the operators $\vc{M}_{f}$, $\vc{T}_{\tau}$ (see \eqref{eqn:signal_ops_def}) one obtains for the latter inner product,
\begin{align*}
 \left\langle \vc{M}_{f} \vc{T}_{\tau} \vc{s}_i , \vc{M}_{f^\prime} \vc{T}_{\tau^\prime} \vc{s}_j \right\rangle
  &= \sum_{k=1}^{N_t} \overline{[\vc{M}_{f} \vc{T}_{\tau} \vc{s}_i]_k} [\vc{M}_{f^\prime} \vc{T}_{\tau^\prime} \vc{s}_j]_k = \sum_{k=1}^{N_t} \ee{\frac{f^\prime -f}{N_t}(k-1)} \overline{[\vc{s}_i]_{k - \tau}} [\vc{s}_j]_{k - \tau^\prime}  \\
  &= \sum_{a,b=1}^{N_t} \delta^{\sim_{N_t}}_{a-b , \tau^\prime -\tau} \ee{\frac{f^\prime -f}{N_t} (a+\tau -1)} \overline{[\vc{s}_i]_{a}} [\vc{s}_j]_{b} .
\end{align*}
By combining the identities from above one finds
\[
 [ \widetilde{\vc{A}}_{S_{[\beta]}}^* \widetilde{\vc{A}}_{S_{[\beta]}} ]_{\varTheta , \varTheta^\prime} = \sum_{i,j=1}^{N_T} \sum_{a,b=1}^{N_t} \overline{s_{(i,a)}} s_{(j,b)} \underbrace{\delta^{\sim_{N_t}}_{a-b , \tau^\prime -\tau} (N_T N_t)^{-1} \ee{\frac{f^\prime -f}{N_t} (\tau +a-1)} \ee{d_T \Delta_\beta [\beta^\prime (j-1) - \beta (i-1)]}}_{= [\vc{Y}^{(i,a),(j,b)}]_{\varTheta , \varTheta^\prime},\, \text{see \eqref{eqn:def_B_matrix_in_sum_2}}} ,
\]
which shows \eqref{eqn:AA_repr_2}.

\goodbreak
The matrices $\vc{Y}^{(i,a),(j,b)}$ allow for a simple formula for their adjoints, namely
\begin{equation}\label{eqn:Y_adjoint}
 [\vc{Y}^{(i,a),(j,b)}]^* = \vc{Y}^{(j,b),(i,a)} .
\end{equation}

\subsubsection*{The product $\vc{F}^* \vc{F}$}
The matrix $\vc{F}$ consists of the blocks $\vc{Y}^{(i,a),(j,b)}$ given by \eqref{eqn:def_B_matrix_in_sum_2}.
Therefore, $\vc{F}$ is self-adjoint so that $\vc{F}^* \vc{F} = \vc{F}^2$. 
Like $\vc{F}$ also the product $\vc{F}^2$ consists of blocks 
and the block at the $(i,a)$-th (block) row and the $(j,b)$-th (block) column is given by
\begin{equation}\label{eqn:sum_quantity}
 [\vc{F}^2]^{(i,a),(j,b)} = \sum_{r=1}^{N_T} \sum_{q = 1}^{N_t} \vc{Y}^{(i,a),(r,q)} \vc{Y}^{(r,q),(j,b)} .
\end{equation}
Recalling \eqref{eqn:def_B_matrix_in_sum_2}, the appearing summands $\vc{Y}^{(i,a),(r,q)} \vc{Y}^{(r,q),(j,b)}$ are given entrywise by
\begin{align}
&  [ \vc{Y}^{(i,a),(r,q)} \vc{Y}^{(r,q),(j,b)} ]_{\varTheta , \varTheta^\prime}
  = \sum_{\tilde\varTheta \in S_{[\beta]}} \vc{Y}_{\varTheta,\tilde\varTheta}^{(i,a),(r,q)} \vc{Y}_{\tilde\varTheta,\varTheta^\prime}^{(r,q),(j,b)} \notag \\
 & = (N_T N_t)^{-2} \sum_{\tilde\varTheta \in S_{[\beta]}} \delta^{\sim_{N_t}}_{a-q , \tilde\tau -\tau} \delta^{\sim_{N_t}}_{q-b , \tau^\prime -\tilde\tau} \ee{\frac{\tilde f -f}{N_t} (\tau +a-1)} \ee{\frac{f^\prime - \tilde f}{N_t} (\tilde\tau +q-1)} \ee{d_T \Delta_\beta [\beta^\prime (j-1) - \beta (i-1)]} \notag  \\
  &= (N_T N_t)^{-2} \delta^{\sim_{N_t}}_{\tau^\prime + b , \tau + a} |S_{[\beta]}^{\tau + a - q}| \ee{\frac{f'-f}{N_t} (a+\tau -1)} \ee{d_T \Delta_\beta [\beta^\prime (j-1)-\beta (i-1)]} \label{eqn:YY_general_prod} .
\end{align}
Combining this with \eqref{eqn:sum_quantity} yields
\begin{align*}
 [\vc{F}\vc{F}]^{(i,a),(j,b)}_{\varTheta ,\varTheta^\prime}
  &= \sum_{r=1}^{N_T} \sum_{q = 1}^{N_t} \delta^{\sim_{N_t}}_{\tau^\prime + b , \tau + a} (N_T N_t)^{-2} |S_{[\beta]}^{\tau + a - q}| \ee{\frac{f'-f}{N_t} (a+\tau -1)} \ee{d_T \Delta_\beta [\beta^\prime (j-1)-\beta (i-1)]}  \\
  &= \delta^{\sim_{N_t}}_{\tau^\prime + b , \tau + a} N_T^{-1} N_t^{-2} |S_{[\beta]}| \ee{\frac{f'-f}{N_t} (a+\tau -1)} \ee{d_T \Delta_\beta [\beta^\prime (j-1)-\beta (i-1)]} .
\end{align*}
\subsubsection*{Properties of the matrices $\vc{Y}^{(i,a),(j,b)}$}
The proof of Proposition \ref{prop:opnorm_bound} uses the identities
\begin{equation}\label{eqn:matrix_sums}
\sum_{i,j=1}^{N_T} \sum_{a,b=1}^{N_t} [\vc{Y}^{(i,a),(j,b)}]^* \vc{Y}^{(i,a),(j,b)} = \sum_{i,j=1}^{N_T} \sum_{a,b=1}^{N_t} \vc{Y}^{(i,a),(j,b)} [\vc{Y}^{(i,a),(j,b)}]^* = \frac{|S_{[\beta]}|}{N_t} \ID .
\end{equation}
Due to \eqref{eqn:Y_adjoint}, and by plugging in the identity we used in the second step of \eqref{eqn:YY_general_prod}, the second sum is given entrywise by
\begin{align*}
&  \sum_{i,j=1}^{N_T} \sum_{a,b=1}^{N_t} [\vc{Y}^{(i,a),(j,b)} \vc{Y}^{(j,b),(i,a)}]_{\varTheta,\varTheta^\prime}
  = \sum_{i,j=1}^{N_T} \sum_{a,b=1}^{N_t} \delta_{\tau^\prime , \tau} \frac{|S_{[\beta]}^{\tau + a - b}|}{(N_T N_t)^{2}} \ee{\frac{f'-f}{N_t} (\tau +a-1)} \ee{d_T \Delta_\beta (\beta^\prime -\beta ) (i-1)}  \\
  &= \delta_{\tau^\prime , \tau} \frac{N_T |S_{[\beta]}|}{(N_T N_t)^{2}} \ee{\frac{f'-f}{N_t} \tau} \sum_{a=1}^{N_t} \ee{\frac{f'-f}{N_t} (a-1)} \sum_{i=1}^{N_T} \ee{d_T \Delta_\beta (\beta^\prime -\beta ) (i-1)}  
  = \delta_{\varTheta , \varTheta^\prime} \frac{|S_{[\beta]}|}{N_t} ,
\end{align*}
which establishes the second equality in \eqref{eqn:matrix_sums}.
The first equality follows due to symmetry. 

\section{Basics from probability theory}\label{sec:tools}
A complex-valued random variable $\xi$ is standard complex Gaussian iff it has (complex) density $\frac{1}{\pi} e^{-|\xi|^2}$, or, equivalently, $\xi$ can be written as $\xi = x + \imath y$, where $x,y \sim N (0,1/2)$ are independent standard Gaussian random variables.
More generally, a mean-zero complex Gaussian random variable with variance $\sigma^2$ is of the form $\sigma \xi$, where $\xi$ is a standard complex Gaussian.

\begin{lemma}\label{lem:complexGaussian}
 For a standard complex Gaussian random variable $\xi$ it holds
 \[
  \P ( | \xi | \geq t ) \leq e^{-t^2} .
 \]
\end{lemma}

For a standard complex Gaussian random vector $\vc{\xi}$ (having independent, standard complex Gaussian entries) and a (deterministic) complex vector $\vc{a}$ of the same dimension, the random variable $z := \langle \vc{a} , \vc{\xi} \rangle$ is mean-zero complex Gaussian with variance $\|\vc{a}\|_2^2$.
This implies the next statement.

\begin{lemma}\label{lem:inner_normal_bound}
 For a standard complex Gaussian random vector $\vc{\xi}$ and a complex vector $\vc{a}$ of the same dimension it holds
 \[
  \P ( | \langle \vc{a} , \vc{\xi} \rangle | \geq t ) = \P ( | \xi | \geq t / \| \vc{a} \|_2 ) \leq e^{-t^2 / \| \vc{a} \|_2^2} .
 \]
\end{lemma}

For a $2n$-dimensional standard Gaussian random vector $\vc{g}$ we have, due to \cite[(8.89)]{Foucart2013},
\[
 \P ( \| \vc{g} \|_2 \geq \sqrt{2n} + t ) \leq e^{-t^2 / 2} .
\]
Since an $n$-dimensional standard complex Gaussian random vector $\vc{\xi}$ can be considered as a (real-valued) $2n$-dimensional standard Gaussian random vector $\vc{g}$ with independent entries from $\NN (0,1/2)$, we have the following lemma.

\begin{lemma}\label{lem:gaus_vec_tail}
 For an $n$-dimensional standard complex Gaussian random vector $\vc{\xi}$ it holds
 \[
  \P ( \| \vc{\xi} \|_2 \geq \sqrt{n} + t ) = \P ( \| 2^{-1/2} \vc{g} \|_2 \geq \sqrt{n} + t ) = \P ( \| \vc{g} \|_2 \geq \sqrt{2n} + \sqrt{2} t ) \leq e^{-t^2} .
 \]
\end{lemma}

Finally, the following lemma states a Hoeffding-type inequality for Steinhaus sequences.
Recall that a Steinhaus sequence is a sequence of independent random variables which are all distributed uniformly on the complex unit circle $\{ z \in \C \midcol |z|=1 \}$.

\begin{lemma}[\mbox{\cite[Cor. 8.10]{Foucart2013}}]\label{lem:hoeffding}
 Let $\vc{a} \in \C^L$ and $\vc{\epsilon} = (\epsilon_1,\ldots,\epsilon_L)$ be a Steinhaus sequence.
 Then
 \[
  \P ( | \langle \vc{a} , \vc{\epsilon} \rangle | \geq u \|\vc{a}\|_2) \leq 2 e^{-u^2 / 2} .
 \]
\end{lemma}

\bibliographystyle{acm}
\bibliography{Sparse-MIMO-Radar}

\end{document}

%% file: latex_inc.tex
\DeclareMathOperator{\sgn}{sgn}

\DeclareMathOperator{\supp}{supp}

\DeclareMathOperator{\conv}{conv}

\DeclareMathOperator{\id}{Id}

\def\C{\mathbb{C}}

\def\E{\mathbb{E}}

\def\N{\mathbb{N}}

\def\P{\mathbb{P}}

\def\R{\mathbb{R}}

\def\Z{\mathbb{Z}}

\def\AA{\mathcal{A}}

\def\CC{\mathcal{C}}
\def\DD{\mathcal{D}}
\def\EE{\mathcal{E}}

\def\GG{\mathcal{G}}

\def\II{\mathcal{I}}

\def\NN{\mathcal{N}}

\def\TT{\mathcal{T}}
\def\UU{\mathcal{U}}

\def\XX{\mathcal{X}}
\def\YY{\mathcal{Y}}